\documentclass[a4paper,USenglish]{article}
\usepackage{microtype}
\usepackage{textcomp, amsmath, amsthm, amssymb}
\usepackage{bm}
\usepackage{float}
\usepackage[]{inputenc}
\usepackage[]{graphicx}
\usepackage{authblk}
\usepackage{subcaption}
\usepackage{multirow}
\usepackage{mdframed}
\usepackage{algpseudocode}
\usepackage{algorithmicx}
\usepackage{algorithm}
\usepackage{thmtools}
\usepackage{nameref,hyperref,xcolor}
\hypersetup{colorlinks, linkcolor={blue!50!black}, citecolor = {green!50!black}}
\title{Shortest $k$-Disjoint Paths via Determinants}
\author[1,2]{Samir Datta}
\author[1,3]{Siddharth Iyer}
\author{Raghav Kulkarni}
\author{Anish Mukherjee}
\affil{Chennai Mathematical Institute, India\\   \texttt{sdatta@cmi.ac.in,kulraghav@gmail.com,anish@cmi.ac.in}}
\affil[2]{UMI ReLaX}
\affil[3]{Birla Institute of Technology and Science - Pilani, India\\  \texttt{sviyer97@gmail.com}}
\usepackage[top=1in,right=1in,bottom=1in,left=1in]{geometry}

 \newtheorem{theorem}{Theorem}
 \newtheorem{lemma}[theorem]{Lemma}

 \newtheorem{definition}[theorem]{Definition}
\newtheorem{observation}[theorem]{Observation}
 
\newtheorem{proposition}[theorem]{Proposition}
 
 \newtheorem{remark}[theorem]{Remark}
\newtheorem{claim}[theorem]{Claim}
\newtheorem{fact}{Fact}

\newcommand{\Log}{\mbox{{\sf L}}}

\newcommand{\Pt}{\mbox{{\sf P}}}
\newcommand{\qP}{\mbox{{\sf quasi-P}}}
\newcommand{\NP}{\mbox{{\sf NP}}}
\newcommand{\Gl}{\mbox{{\sf GapL}}}

\newcommand{\NC}{\mbox{{\sf NC}}}
\newcommand{\bfNC}{\mbox{{\textbf{\textsf{NC}}}}}
\newcommand{\bfRNC}{\mbox{{\textbf{\textsf{RNC}}}}}

\newcommand{\RNC}{\mbox{{\sf RNC}}}

\newcommand{\TCz}{\mbox{{\sf TC}$^0$}}

\newcommand{\NCo}{\mbox{{\sf NC}$^1$}}

\newcommand{\NCt}{\mbox{{\sf NC}$^2$}}

\newcommand{\AxisCross}{\mbox{{\sf AxisCross}}}

\newcommand{\CD}{\mbox{{\sf CD}}}

\newcommand{\Offset}{\sf O}

\providecommand{\nc}{\newcommand}

\newcommand  {\myclass} [1]  {\ensuremath{\textsc{#1}}}

\nc{\dynac}{\ensuremath{\myclass{DynAC}^0}\xspace}
\nc{\adom}{\ensuremath{\text{adom}}}
\nc{\subs}{\subseteq}

\newcommand{\ang}[1]{\left\langle #1 \right\rangle}
\newcommand{\mylen}[1]{{\mathbf{len}}{(#1)}}

\makeatletter 
\newcommand{\uf}[4]{
  \@ifmtarg{#4}{
    \ensuremath{\phi^{#1}_{#2}(#3)}
   }{
    \ensuremath{\phi^{#1}_{#2}(#3; #4)}
  }
}

\newcommand{\M}{$\mathbf{M}$}
\newcommand{\comment}[1]{}
\newcommand{\mcomment}[3]{}

\nc{\sdm}[1]{\mcomment{SD}{Samir}{#1}}
\nc{\rkm}[1]{\mcomment{RK}{Raghav}{#1}}
\nc{\amm}[1]{\mcomment{AM}{Anish}{#1}}
\nc{\sm}[1]{\mcomment{SI}{Siddharth}{#1}}
\begin{document}
\maketitle
\begin{abstract}
The well-known $k$-disjoint path problem ($k$-DPP) asks for pairwise 
vertex-disjoint paths between $k$ specified pairs of vertices  $(s_i, t_i)$ 
in a given graph, if they exist. The decision version of the shortest $k$-DPP 
asks for the length of the shortest (in terms of total length) such paths. 
Similarly the search and counting versions ask for one such and the number of such shortest set of paths, respectively.

We restrict attention to the shortest $k$-DPP instances on
undirected planar graphs where all sources and sinks lie on a single face 
    or on a pair of faces. 
We provide efficient sequential and parallel algorithms for the search 
versions of the problem answering one of the main open questions raised 
by Colin de Verdiere and Schrijver \cite{VS11} for the general one-face
problem. We do so by providing a randomised \NCt\ algorithm along with an 
$O(n^{\omega})$ time randomised sequential algorithm. We also obtain
deterministic algorithms with similar resource bounds for the counting and search versions.

In contrast, previously, only the sequential complexity of 
decision and search versions of the ``well-ordered'' case has
been studied. For the one-face case, sequential versions of our 
routines have better running times for constantly many terminals.
In addition, the earlier best known sequential algorithms (e.g. 
Borradaile et al. \cite{BNZ}) were randomised while ours are also
deterministic.

The algorithms are based on a bijection between a shortest $k$-tuple of 
disjoint paths in the given graph and cycle covers in a related digraph. 
This allows us to non-trivially modify established techniques relating
counting cycle covers to the determinant. We further need to do a controlled
inclusion-exclusion to produce a polynomial sum of determinants such that 
all ``bad'' cycle covers cancel out in the sum allowing us to count ``good''
cycle covers.
\end{abstract}
\newpage
\section{Introduction}\label{sec:intro}
\subsection{The $k$-disjoint path problem}
The $k$-Disjoint Path Problem, denoted by $k$-DPP, is a well-studied problem in algorithmic graph theory with many applications in transportation networks, VLSI-design and most notably in the algorithmic graph minor theory (see for instance \cite{KS} and references therein). The $k$-DPP can be formally defined as follows:
Given a (directed/undirected) graph $G = (V, E)$ together with $k$ specified pairs of terminal vertices $(s_i, t_i)$ for $i \in [k]$, find $k$ pairwise vertex-disjoint paths $P_i$ from $s_i$ to $t_i$, if they exist. One may similarly define an edge-disjoint variant (EDPP) of the problem. 
We will mainly focus on the vertex-disjoint variant in this paper though several of our results are translated to the edge-disjoint version. The Shortest $k$-DPP asks to find $k$ pairwise vertex-disjoint paths of minimum total length. We consider the following  variants of Shortest $k$-DPP\sm{Maybe we can also talk about our hardness result here, at least for the one-face case}:
\begin{enumerate}
\item Decision: given $w$, decide if there is a set of $k$-disjoint paths of length at most $w$.
\item Construction/Search: construct one set of shortest $k$-disjoint paths. 
\item Counting: count the number of shortest $k$-disjoint paths.
\end{enumerate}
\subsection{Finding $k$-disjoint paths : Historical overview}
The existence as well as construction versions of $k$-DPP are well-studied in general as well as planar graphs. The problem in general directed graphs is \NP -hard even for $k = 2$ \cite{FHW80}. DPP is one of Karp's \NP-hard problems \cite{RK} and remains so when restricted to undirected planar graphs \cite{L75} and \cite{MiddendorfP93} extends this to EDPP as well. In fact, EDPP remains \NP-hard even on planar undirected graphs when all the terminals lie on a single face~\cite{Schwarzler09}.The problem of finding two disjoint paths, one of which is a shortest path, is also \NP-hard \cite{Eil}.

The existence of a One/Two-face $k$-DPP was studied in \cite{RobertsonS86a} as part of the celebrated \emph{Graph Minors} series. This was extended (for fixed $k$) to graphs on a surface \cite{RobertsonS88} and general undirected graphs \cite{RS} in later publications in the same series \cite{RS}. A solution to this problem was central to the Graph Minors Project and adds to the importance of the corresponding optimization version. Suzuki et al. \cite{SAN90} gave linear time\sm{Also mention linear time algorithm for disjoint trees and paths by Reed, Robertson, Schjriver and Seymour} and $O(n \log n)$ time algorithms for the One-face and Two-face case, respectively and \cite{SYN90} gave \NC\ algorithms for both. In directed graphs, for fixed $k$ polynomial time algorithms are known when the graph is either planar \cite{Sch94} or acyclic \cite{FHW80}.

Though there are recent exciting works on planar restrictions of the problem (e.g. \cite{ChuzhoyKL16}\sm{also cite paper on approximating maxDPP in bounded treewidth graphs}), the One-face or Two-face setting might appear on first-look to be a bit restrictive. However, the One-face setting occurs naturally in the context of routing problems for VLSI circuits where the graph is a two dimensional grid and all the terminals lie on the outer face. Relaxations of the one-face setting become intractable, e.g., ``only all source-terminals on one face'' is hard to even approximate under a reasonable complexity assumption ($\NP \neq \qP$ \cite{ChuzhoyKN16}). 
\subsection{Shortest $k$-DPP : Related work}
The optimization problem is considerably harder. A version of the problem, called \emph{length-bounded} DPP, where the each of the path need to have length bounded by some integer $\ell$. This problem is \NP -hard in the strong sense even in the One-face case for non-fixed $k$ \cite{vd02}. For the shortest $k$-DPP, where we want to minimise the sum of the lengths of the paths, very few instances are known to be solvable in polynomial time. For general undirected graphs, very recently, Bj\"{o}rklund and Husfeldt \cite{BH14} have shown that shortest $2$-DPP admits a randomised polynomial time algorithm. The deterministic polynomial time bound for the same – to this date – remains an intriguing open question.

For planar graphs, Colin de Verdi\`{e}re and Schrijver \cite{VS11} and Kobayashi and Sommer \cite{KS} give polynomial time algorithms for shortest $k$-DPP in some special cases. \cite{VS11} gives an $O(kn\log{n})$ time algorithm works when the source and sink vertices are incident on different faces of the graph and allows $k$ to be a part of the input. \cite{KS} gives $O(n^4\log{n})$ and $O(n^3\log{n})$ time algorithms when the terminal vertices are on one face for $k \le 3$ or on two faces for $k = 2$, respectively. For arbitrary $k$, linear time algorithm is known for bounded tree-width graphs \cite{Sch94a}. Polynomial time algorithms are also known through reducing the problems to the minimum cost flow problem when all the sources (or sinks) coincides or when the terminal vertices are in parallel order \cite{vd02}. 

In \cite{VS11} authors ask about the existence of a polynomial time algorithm provided all the terminals are on a common face, for which we give an efficient deterministic algorithm for any fixed $k$. The only progress on this was made in \cite{BNZ} where an $O(kn^5)$ time randomised algorithm is presented when corresponding sources and sinks are in series on the boundary of a common face. All the previous one-face planar results are strictly more restrictive or orthogonal to our setting and our sequential algorithms are more efficient (for fixed $k$). We are able to tackle the counting version that is typically harder than the decision version. 
Also, to the best of our knowledge, none of the previous works have addressed the parallel complexity of these problems.
\subsection{Our results and techniques}\label{subsec:results}
We resolve an open question in \cite{VS11} and also provide a simpler algorithm for the problem considered in \cite{VS11} while ensuring that our algorithms are efficiently parallelizable:
\begin{theorem}\label{thm:main}
Given an undirected planar graph $G$ with either $k$ pairs of source and sink vertices lying on one face or $k$ source vertices on one face and $k$ sink vertices on another, then
\begin{enumerate}
\item count of all shortest $k$-disjoint paths between the terminals can be found in \NCt,
\item the length of a shortest set $k$-disjoint paths between the terminals can be found in \NCt, and
\item a shortest set of $k$-disjoint paths between the terminals can be found in $\Pt\cap\RNC$.
\end{enumerate}
\end{theorem}
Our algorithms extend to a variant of the edge-disjoint version of the problem (for decision and search) by known reductions to the vertex disjoint case (see Lemma~\ref{lem:reduction} in Section~\ref{sec:edpp}) and for weighted graphs where each edge is assigned a weight which is polynomially bounded in the number of vertices. We obtain running times independent of $k$ when the terminal vertices on the faces are in parallel order. We summarize our main results in Table~\ref{tab:res}.
\begin{table}
\centering
\begin{tabular}{|c|c|c|c|c|}\hline
\multirow{2}{*}{Problem} & \multirow{2}{*}{Variant}& \multicolumn{2}{c|}{Sequential} & \multirow{2}{*}{Parallel}\\\cline{3-4}
 &	& Deterministic      & Randomised &  \\ \cline{1-5}
\multirow{3}{*}{One-face General}&Decision & $\mathbf{4^kn^{\omega}}$    &   & \bfNC   \\
 & Counting  & $\mathbf{4^kn^{\omega+1}}$    &   & \bfNC   \\
& Search & $\mathbf{4^kn^{\omega+2}}$    & $\mathbf{4^kn^{\omega}}$  & \bfRNC  \\\cline{1-5}
\multirow{3}{*}{Two-face Parallel}&Decision & $\mathbf{kn^{\omega}}$    &   & \bfNC   \\
 & Counting  & $\mathbf{kn^{\omega+1}}$    &   & \bfNC   \\
& Search & $\mathbf{kn^{\omega+2}}$ ($kn \log n$\cite{VS11})  & $\mathbf{kn^{\omega}}$  &  \bfRNC   \\\cline{1-5}
\end{tabular}
\caption{Summary of Results. The dependence on $k$ and $n$ of our results (in $\mathbf{bold}$) is emphasized. Note that $\omega$ is the matrix multiplication constant.}\label{tab:res}
\end{table}
The proof of Point 1 in Theorem~\ref{thm:main} depends on four ideas:
\begin{enumerate}
\item[A.] An injection from $k$ disjoint paths to cycle covers in a 
related graph for the general case.
\item[B.] The injection above reduces to a bijection in the parallel case. 
(Lemma~\ref{lem:bij})
\item[C.] An identity involving telescoping sums to simplify the count of 
    $k$-disjoint paths (Lemma~\ref{lem:cancel})
\item[D.] A pruning of the cycle covers in the Two-face case based on 
    topological considerations.
\end{enumerate}
Point 2 is an immediate consequence of Point 1.
For the \Pt-bound in Point 3 we use a greedy method with a counting oracle from
Point 1. For the \RNC-algorithm we instead use isolation \emph{a la} \cite{MVV}. We sketch these ideas in more detail below.
\subsubsection*{Proof Sketch}  
Throughout the following sketch we talk about pairings which are essentially a collection of $k$ source-sink pairs, though not necessarily the same one which was specified in the input. We refer to this input pairing by $M_0$.
\begin{enumerate}
\item \textbf{One Face Case.} We first convert the given undirected planar graph into a directed one such
that each set of disjoint paths between the source-sink pairs in $M_0$
corresponds to directed cycle covers and this correspondence preserves
weights (Lemma~\ref{lem:inj}).
In this process, we might introduce ``bad'' cycle covers corresponding to 
pairings of terminals which are not required and they need to be cancelled out. 
Each ``bad'' cycle cover which was included, can be mapped to a unique pairing, say $M_1$.
Since the ``bad'' cycle cover occurs in $M_0$ as well as $M_1$ we can cancel it out 
by adding or subtracting the determinant of $M_1$ from $M_0$. 
However, $M_1$ can introduce further ``bad'' cycle covers which again need to be cancelled. 
We show that all the ``bad'' cycle covers like this can be cancelled 
by adding or subtracting determinants exactly like in an inclusion-exclusion formula over a DAG (Lemma~\ref{lem:cancel}).
This process terminates with the so called ``parallel'' pairings (where
the weight-preserving correspondence between $k$-disjoint paths and 
cycle covers with $k$ non-trivial cycles is a bijection) (Lemma~\ref{lem:bij}).
\item \textbf{Two Face Case.} The inclusion-exclusion formula exploited the topology of the one face case which is not present in the two face case. Here, this approach breaks down as the pairings can not be put together as a DAG. We resolve this for a special case when all sources are on one face and all sinks are on the other by using a topological artifice to prune out pairings which cause cycles. 
\item \textbf{Counting.} The cycle covers in a graph can be counted by a determinant - more
precisely, we have a univariate polynomial which is the determinant of some
matrix  such that every cycle cover corresponds with one monomial in the 
determinant expansion. 
Since the ``bad'' cycle covers cancel out in the inclusion-exclusion,
the coefficient of the least degree term gives the correct count of the 
shortest cycle covers in $M_0$ which can then be extracted out by interpolation.
    For the two face case, we need the number of cycle covers with a 
    certain winding number modulo $k$. This can be read off from the monomial with the appropriate exponent 
    in the determinant polynomial.
\item \textbf{Search.} Using standard isolation techniques \cite{MVV}
we can construct a set of shortest $k$-disjoint paths in \RNC.
Similarly, a greedy strategy with the counting procedure as an oracle yields the witness in \Pt.
Since we essentially reduce counting shortest $k$-DPP solutions to computing
$O(4^k)$ determinants, we can do the counting in $O(4^kn^{\omega+1})$ time
where $\omega$ is the matrix multiplication constant. 
\end{enumerate}
\subsubsection*{Main Technical Contribution}
Our main technical ingredient here is the Cancellation Lemma~\ref{lem:cancel}
that makes it possible to reduce the count of disjoint paths
to signed counts over a larger set in such a way that the spurious terms cancel out.
This reduces the count of disjoint paths to the determinant.
To the best of our knowledge this is the first time a variant of the disjoint path problem 
has been reduced to the determinant, a parallelizable quantity (in contrast \cite{BH14} 
reduce $2$-DPP to the Permanent modulo $4$ for which no parallel algorithm is known).
\subsection{Organization}
We recall some preliminaries in Section~\ref{sec:prelims} and
describe the connection between $k$-disjoint paths and determinant in Section~\ref{sec:det}.
Section~\ref{sec:onefacepar} solves the parallel One-face case.
In Section~\ref{sec:onefacegen} we discuss the general One-face case and in Section~\ref{sec:twoface} the parallel Two-face case. 
In Section~\ref{sec:mainproof} we give the proof of Theorem~\ref{thm:main}. 
We extend our results for shortest $k$-DPP to a variant of shortest $k$-EDPP in Section~\ref{sec:edpp}.
We conclude in Section~\ref{sec:open} with some open ends.
\section{Preliminaries}\label{sec:prelims}
An \emph{embedding} of a graph $G = (V, E)$ into the plane is a mapping from $V$ to different points of $\mathbb{R}^2$, and from $E$ to internally disjoint simple curves in $\mathbb{R}^2$ such that the endpoints of the image of $(u,v) \in E$ are the images of vertices $u ,v \in V$. If such an embedding exists then $G$ is planar. The faces of an embedded planar graph $G$ are the maximal connected components of $\mathbb{R}^2$ that are disjoint from the image of $G$. We can find a planar embedding in \Log\ using \cite{AM04,DP11}.
Our proofs go through by reducing the problems to counting/isolating cycle covers.
Since the determinant of the adjacency matrix of a graph is the signed sum of its cycle
covers, we can count the lightest cycle covers by ensuring that all such cycle covers get the
same sign. Similarly, isolating one lightest cycle cover enables us to extract it via determinant computations.

We note the following seemingly innocuous but important:
\begin{fact}[see e.g. \cite{MV}]\label{fact:innocuous}
The sign of a permutation $\pi \in S_n$ equals $(-1)^{n + c}$ where $c$ is
the number of cycles in $\pi$.
\end{fact}
\begin{figure}
\centering
\subcaptionbox{Parallel}
  {\includegraphics[width=.22\linewidth]{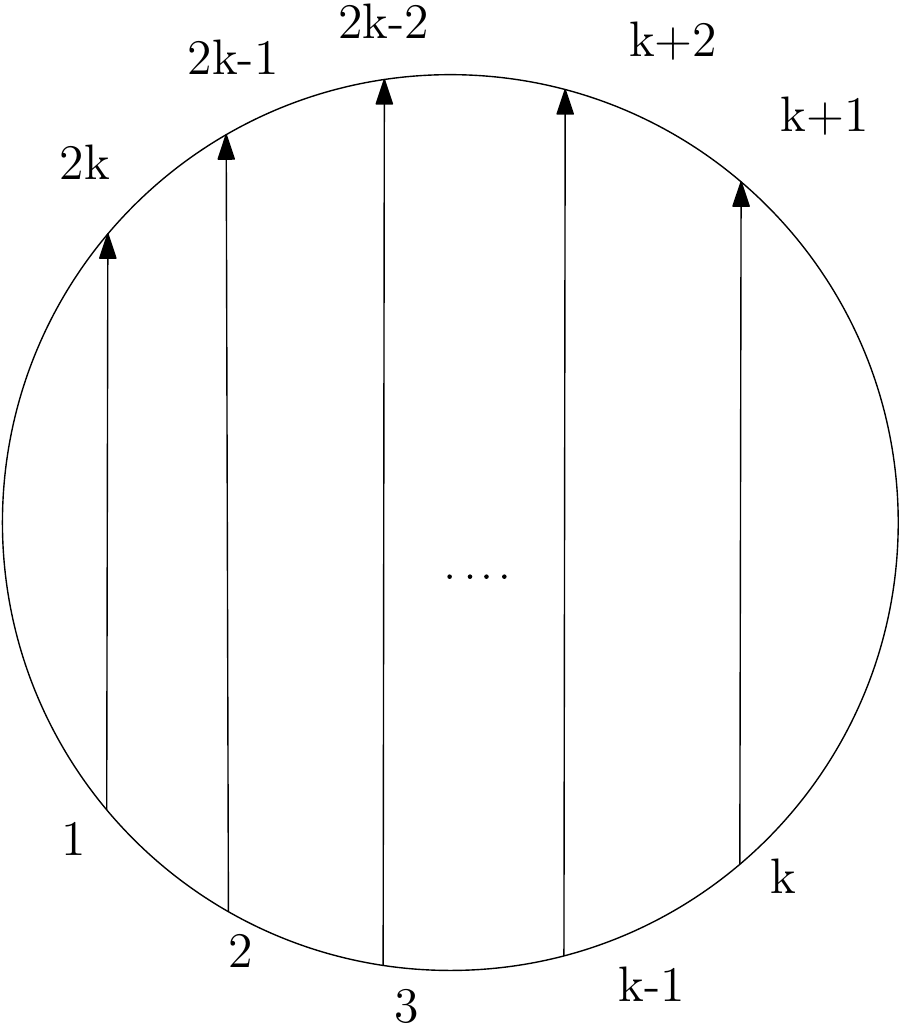}}
\hfill
\subcaptionbox{Serial}
  {\includegraphics[width=.259\linewidth]{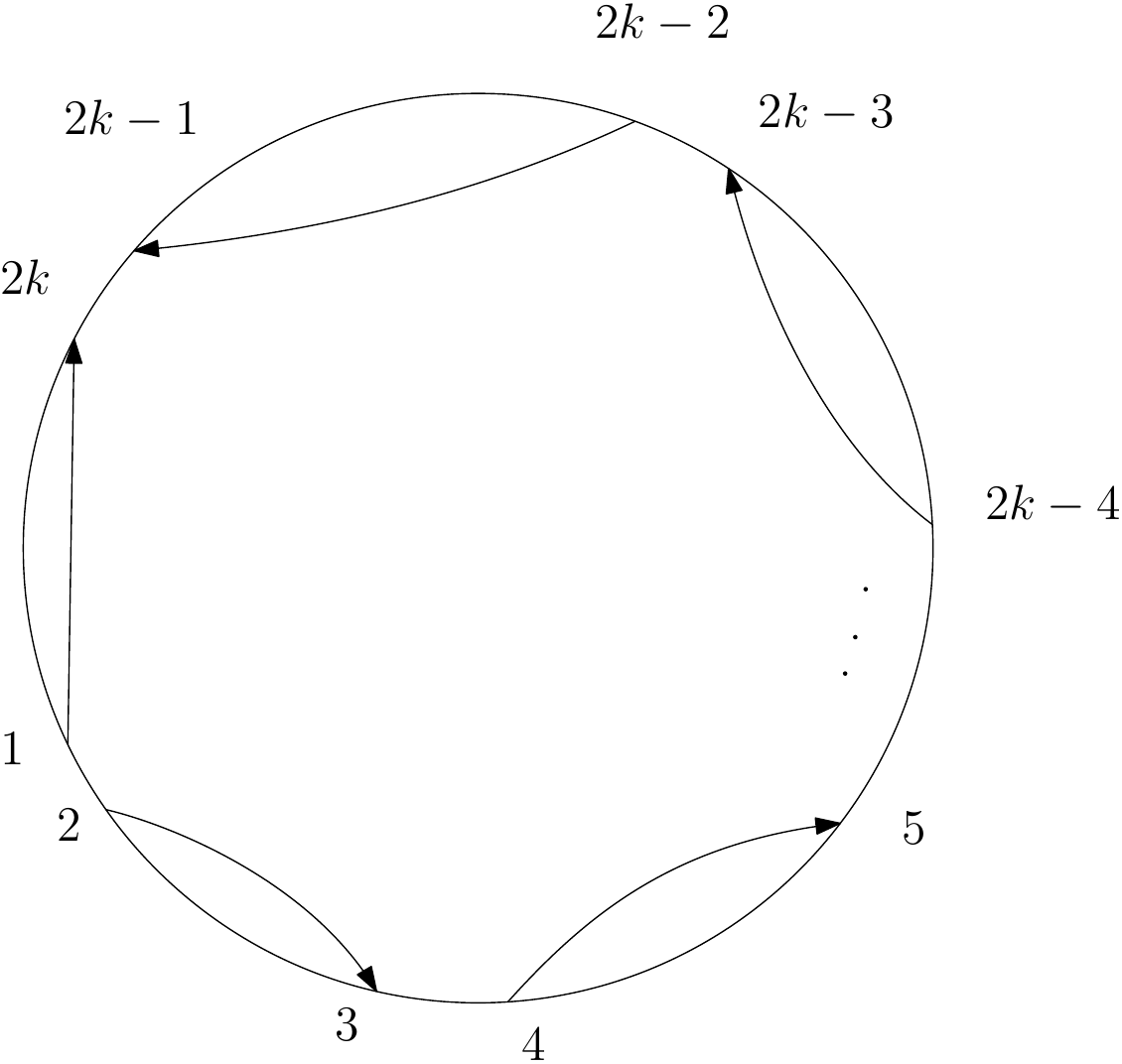}}
\hfill
  \subcaptionbox{General}
  {\includegraphics[width=.22\linewidth]{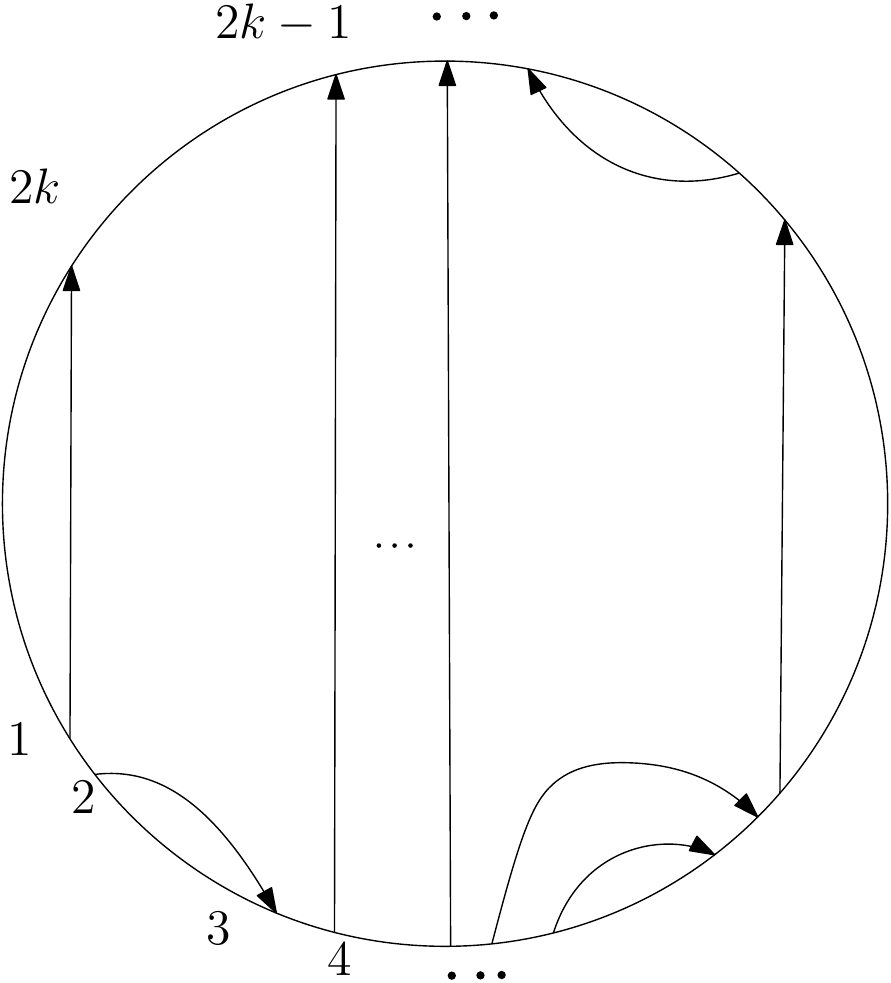}}
\caption{(a) Parallel (b) Serial (c) General Terminal Pairings}
\label{fig:configs}
\end{figure}
Let $G$ be a plane graph. We say that a set of $k$ terminal pairs 
$\{(s_i,t_i): i \in [k]\}$ (so called \emph{pairing}) is \emph{One-face} if the terminals 
all occur on a single face $F$. They are in parallel order if the pairs occur in the order 
$s_1,s_2,\ldots,s_k,t_k,\ldots,t_2,t_1$ on the facial boundary
and in serial order if they occur in the order $s_1,t_1,s_2,t_2,\ldots,s_k,t_k$. 
Otherwise the pairing is said to be in general order. 
If all the $k$ terminal pairs occur on two-faces $F_1$ and $F_2$, we call it \emph{Two-face}.
Here they are in parallel order if the sources $s_1,s_2,\ldots,s_k$ occur on one face and all the 
sinks $t_1,t_2,\ldots,t_k,$ are on another.
They are in serial order if on each of the two faces they occur in serial order (in the sense of One-face).
\section{Disjoint paths, cycle covers and determinant}\label{sec:det}
We first describe a basic graph modification step using which we can show
connections between cycle covers and shortest $k$-DPP. In the rest of the
paper, we will first perform the modification before applying our algorithms.
\subparagraph*{Modification Step.} 
Let $G$ be an undirected graph with $2k$ vertices
called as \emph{terminals}. The terminals are paired together into
$k$ disjoint ordered pairs. We refer to the $i^{th}$ pair as $(s_i,t_i)$, where
$s_i$ is the source and $t_i$ is the sink. We add directed edges
from the sink to source in each terminal pair and refer to them as
\textit{demand edges}. We call a set of $k$ such demand edges as a \text{pairing}. 
We subdivide each demand edge to yield $k$ new vertices. These $3k$ vertices we 
deem as special. Let the resulting graph be $H$. 
Lastly, we add self loops to all non-special vertices and weigh the rest of the
edges of $H$ by $x$. Let $H$ have adjacency matrix $A$. After adding
self loops and weighing edges, the resultant adjacency matrix $B$ can be
written as $D + xA$ where $D$ is the diagonal matrix with $1$'s for non-special
vertices and zeroes for special ones.

This is the weighted adjacency matrix of the graph with all original
edges of $H$ getting weight $x$ and with unit weight self loops on
non-special vertices.
There is a bijection between cycle covers in the graph and monomials
in the determinant $det(D + xA)$. Each cycle cover in turn consists
of disjoint cycles which are one of three types:
\begin{enumerate}
\item consisting alternately of paths between two terminals and demand edges.
\item a non-trivial cycle avoiding all terminals.
\item a trivial cycle i.e. a self loop.
\end{enumerate}
Thus every cycle cover contains a set of $k$ disjoint paths.
Further any collection of $k$ disjoint paths between the terminals (not
necessarily in the specified pairing) can be extended in at least one 
way to a cycle cover of the above type. Thus the set of all monomials
of the determinant are in bijection with the set of all $k$ disjoint paths.

Finally we have extensions of ``good'' $k$-disjoint paths 
(which are between a designated set of pairs of terminals), 
which are in bijection with a \emph{subset} of all cycle covers.
We call the corresponding set of cycle covers \emph{good cycle covers}.
This bijection carries over to some monomials (the so called 
\emph{good monomials}) of the determinant. Thus we obtain the following:
\begin{lemma}\label{lem:inj}
Let $B = D + xA$ as above.
The non-zero monomials in $det(B)$ are in bijection with the cycle
covers in the graph with weighted adjacency matrix $B$ and every 
cycle cover is also an extension of a $k$-disjoint path. 
These bijections also apply to the subset of ``good'' $k$-disjoint paths to yield, 
so called good cycle covers and good monomials. Since the bijection preserves the 
degree of a monomial as the length of the cycle cover it is mapped to, the least-order
term in $det(B)$ corresponds to the ``good" shortest $k$-disjoint paths.
\end{lemma}
Let's focus on the terms that correspond to minimum length ``good" cycle covers. 
Then these terms have the same exponent $\ell$, the length of this shortest ``good" cycle cover. 
This is also the least exponent amongst all the ``good" monomials occurring in the determinant. Notice that their sign 
is the same. To see this, consider the sign given by $(-1)^{n+c}$ (see Fact~\ref{fact:innocuous}) 
where $n$ is the number of vertices and $c$ the number of cycles in  the cycle cover. The 
number of non self-loop cycles is $k$, the minimum number of cycles needed to 
cover all the vertices without self loops and equalling the number of source 
sink pairs. Notice that any extra cycles can be replaced by self loops yielding 
a cycle cover of strictly smaller length hence will not figure in the minimum
exponent term. The number of self loops is therefore $n - \ell$. Hence the
total number of cycles is $k + n - \ell$ for each of these terms hence the sign 
is $(-1)^{k - \ell}$ which is independent of the specific shortest cycle cover under consideration.
\begin{lemma}\label{lem:shortGood}
The shortest good cycle covers all have the same sign.
\end{lemma}
Notice that ultimately we want to cancel out all monomials which are not good. 
In the one face case described in Section~\ref{sec:onefacegen} we show how to do this in the Cancellation Lemma~\ref{lem:cancel}. In the two face case, we cannot do this in general but by measuring how paths wind around the faces, we can characterize the cycle covers which we wish to obtain(see Theorem~\ref{thm:equiLenM}).
\section{Disjoint Paths on One-face: The parallel case}\label{sec:onefacepar}
In this section, we consider directed planar graphs where all the 
terminal vertices lie on a single face in the parallel order. Here we exhibit 
a weight preserving bijection between the set of $k$-disjoint paths in
the given graph and the set of cycle covers with exactly $k$ cycles in a 
\emph{modified graph} $G'$. We first modify the given graph as follows:

\subparagraph{Notation and Modification.} Let $G = (V, E)$ be the given directed 
planar graph with $n$ vertices and $m$ edges. Let $s_1, \ldots, s_k$ and 
$t_k, \ldots, t_1$ be the source and sink vertices 
respectively, all occurring on a face $F$ in the order specified above. 
We apply the modification step described in Section~\ref{sec:det} with one exception,
which is that the subdivided demand edges are of (additive) weight $0$. 
Let the modified graph be $G'$ with $n'$ vertices and $m'$ edges where 
$n' = n+k$ and $m'= m+2k$. $G'$ remains planar. Let $A'$ be the adjacency matrix of $G'$.

Recall that a cycle cover is a collection 
of directed vertex-disjoint cycles incident on every vertex in the graph. A
$k$-cycle cover is a cycle cover containing exactly $k$ non-trivial cycles 
(i.e. cycles that are not self-loops). We show the following bijection: 
\begin{lemma}[Parallel Bijection]\label{lem:bij}
There is a weight-preserving bijection between $k$-disjoint paths and 
$k$-cycle covers in the modified graph $G'$. 
\end{lemma}
\begin{proof}
Suppose the graph $G$ contains a set of $k$ disjoint paths. Consider a 
shortest set of $k$-disjoint paths of total length $\ell$. 
There are $k$ disjoint cycles in $G'$ corresponding to the shortest $k$ 
disjoint paths in $G$, using the new paths from $t_i$ to $s_i$ through $r_i$,
inside the face, for each $i \in [k]$. The $n-\ell-k$ vertices which are not covered 
by these $k$ cycles will use the self loops on them, yielding a $k$-cycle 
cover of $G'$. All these cycle covers have the same weight $\ell$.

For the other direction, consider a $k$-cycle cover in $G'$. If each
non-trivial cycle includes exactly one pair $s_i,t_i$
of terminals then we are done. 

Suppose not, then there is a cycle in the cycle 
cover which contains $s_i$ and $t_j$ for some $ 1\le i\neq j \le k$. We further
assume, without loss of generality, that there are no terminals other than
possibly $s_j,t_i$ between $s_i,t_j$ in the direction of traversal of this
cycle, called, say, $C$. Then $C$ must go through the vertices $r_j$ and $s_j$ 
since the only incoming edge incident on $r_j$ starts at $t_j$ and the only
outgoing edge leads to $s_j$. By the same logic $t_i$ and $r_i$ are on the
cycle $C$. Also notice that the vertices $t_i,r_i,s_i$ must occur
consecutively in that order and so must $t_j,r_j,s_j$. Let the $C$ be
$t_i,r_i,s_i,P_{ij},t_j,r_j,s_j,P_{ji},t_i$ where $P_{ij},P_{ji}$ are paths.
Let the face $F$ be $s_i,F_{ij},s_j,F_j,t_j,F_{ji},t_i,F_i,s_i$ where
$F_{ij},F_{ji},F_i,F_j$ are paths made of vertices and edges from $F$.
Since $C$ is simple $P_{ji}$ cannot intersect $P_{ij}$.
\begin{figure}
\centering
 \includegraphics[width=.35\linewidth]{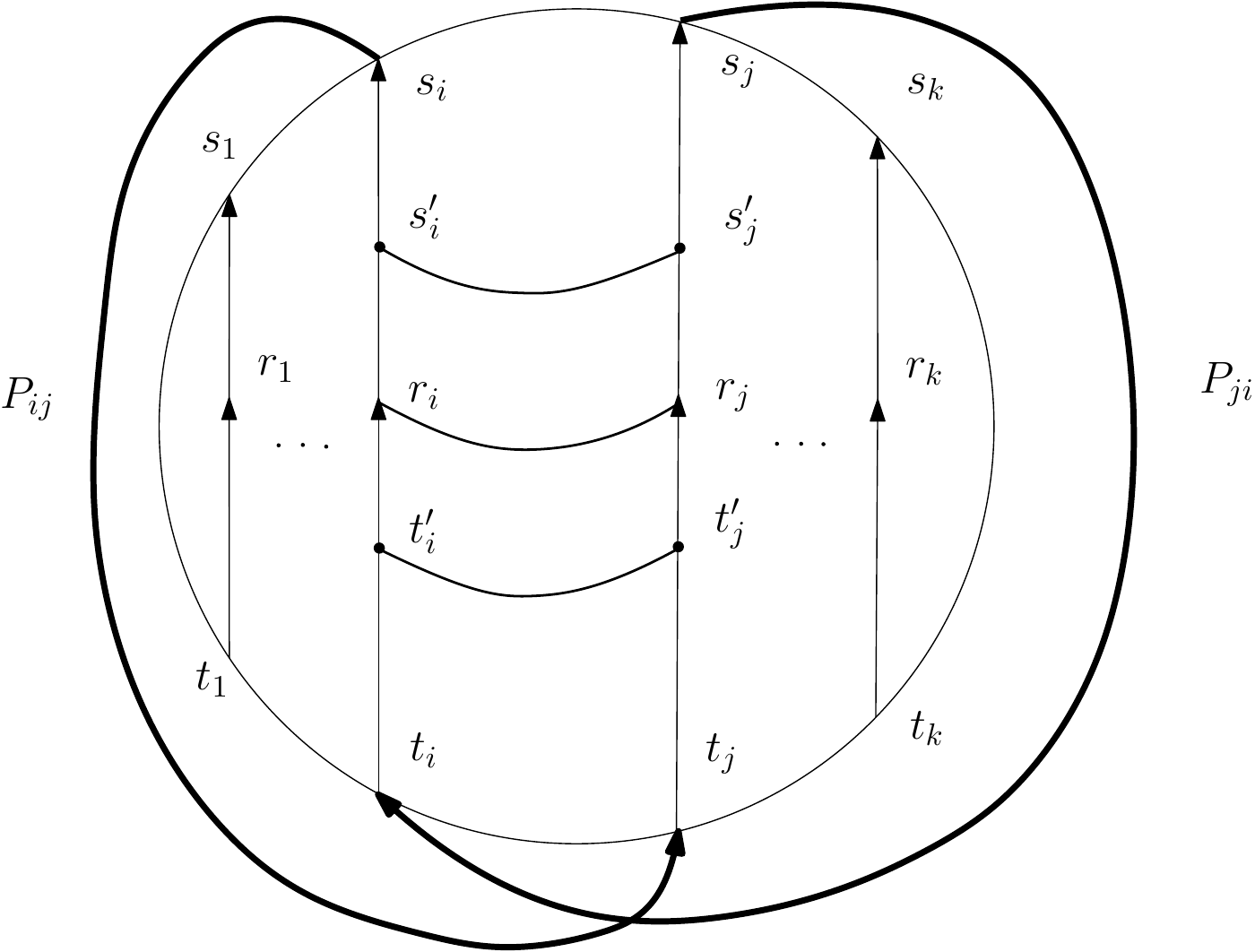}  
  \caption{Parallel Configuration. The bipartite subgraph $\{s'_i, r_j, t'_i\} \cup \{s'_j, r_i, t'_j\}$ gives a $K_{3,3,}$}
  \label{fig:ofp}
\end{figure}
Thus the region inside $F$ bounded by 
$t_i,r_i,s_i,F_{ij},s_j,r_j,t_j,F_{ji},t_i$ does not contain any vertex or edge
from $C$. Thus we can subdivide $(t_i,r_i),(r_i,s_i),(t_j,r_j),(r_j,s_j)$
to introduce vertices $t'_i,s'_i,t'_j,s'_j$ respectively and also the edges
$(t'_i,t'_j),(r_i,r_j),(s'_i,s'_j)$ without affecting the planarity of 
$C \cup F$. But now observe that the complete bipartite graph with 
$\{s'_i,r_j,t'_i\}$ and $\{s'_j,r_i,t'_j\}$ as the two sets of branch vertices
forms a minor of $C \cup F$ augmented with the above vertices and edges.
This contradicts the planarity of $G'$.

As the newly added edges (including the self loops) have weight $0$, the bijection is also weight preserving.
\end{proof}
\section{Disjoint Paths on One Face: The General Case}\label{sec:onefacegen}
In the last section (Section~\ref{sec:onefacepar}) we saw
an important special case - when all demands are in ``parallel'' and now we proceed to the more general case. 
We consider an embedding of an undirected planar graph $G$ with all the terminal vertices 
on a single face in some arbitrary order.
The primary idea is, given graph $G$ to construct a sequence 
of graphs $\mathcal{H}$ so that in the signed sum of the determinants of the graphs in 
$\mathcal{H}$ the uncancelled minimum weight cycle covers are in bijection 
with the shortest $k$-disjoint paths of $G$.
\subparagraph*{Notation and Modification.}
Let $s_1, \ldots, s_k$ and $t_1, \ldots, t_k$ be
the source and sink vertices respectively,
incident on a face $F$ in some arbitrary order. Label the terminals in the
counter clockwise order by $\{1,2,\ldots,2k\}$ and let $\ell(t)$
denote the label of terminal $t$. Consider the graph $G_T$ obtained by
applying the modification step in Section~\ref{sec:det}. A demand edge $(u,v)$ 
is said to be forward if $\ell(u) < \ell(v)$ and reverse otherwise. For any pairing
$M$ if the edges of $M$ are forward we declare the pairing to be in \emph{standard} form.
An even directed cycle is \emph{non trivial} if it has length at least $4$.
\subsection{Pure Cycle Covers}\label{sub:pcc0}
We define \emph{pure cycle covers} of a graph $G$ to be cycle covers $CC$, such that each cycle in $CC$
which contains a terminal also contains the corresponding mate of that terminal and no other terminal.
The mate of a source terminal is the corresponding sink terminal which is specified in the pairing under consideration. 
In the words, in a pure cycle cover no two terminal pairs are part of the same cycle.
Let the graph obtained by deleting all vertices and edges outside $F$ in $G_T$ be $\hat{G}_T$.
If two edges in $\hat{G}_T$ cross then the paths joining corresponding endpoints outside $F$ in $G_T$ will also cross. 
A bit more formally, the following is a consequence of the fact that two cycles in the plane must cross each other 
an even number of times. Notice that the following condition is necessary but not sufficient.  
\begin{observation}\label{obs:pcc1}
Unless $\hat{G}_T$ is outerplanar there is no pure cycle cover in $G$.
\end{observation}
Thus a (single) crossing between edges joining a pair of 
terminals inside $F$ in $G_T$ ensures a crossing between pairs of paths joining 
the same two pairs of terminals outside $F$.
\subsection{Cancelling Bad Cycle Covers}\label{sub:modGen}
\begin{definition}\label{SP} 
Consider two forward demand edges $h_1=(u_1,v_1)$ and $ h_2=(u_2,v_2)$. We say $h_1$ and $h_2$ are in series if either both endpoints of $h_1$ are smaller than both the endpoints of $h_2$ or vice-versa. If however, the sources of $h_1$ and $h_2$ are smaller than the corresponding sinks then the demands could be in parallel or interlacing with each other as follows. 
\begin{enumerate}
\item Parallel: either $\ell(u_1)<\ell(u_2)<\ell(v_2)<\ell(v_1)$ or $\ell(u_2)<\ell(u_1)<\ell(v_1)<\ell(v_2)$.
\item Interlacing: either $\ell(u_1)<\ell(u_2)<\ell(v_1)<\ell(v_2)$ or $\ell(u_2)<\ell(u_1)<\ell(v_2)<\ell(v_1)$.
\end{enumerate} 
\end{definition}
\begin{definition}\label{def:compat} 
An ordered pair $\ang{M,M'}$ of pairings is compatible if, when we direct $M$ in the standard form then there is a way to give directions to $M'$ such that the union of the two directed edge sets forms a set of directed cycles.
\end{definition}
\begin{figure}
\centering
  \subcaptionbox{Compatible Pairings}
  {\includegraphics[width=.27\linewidth]{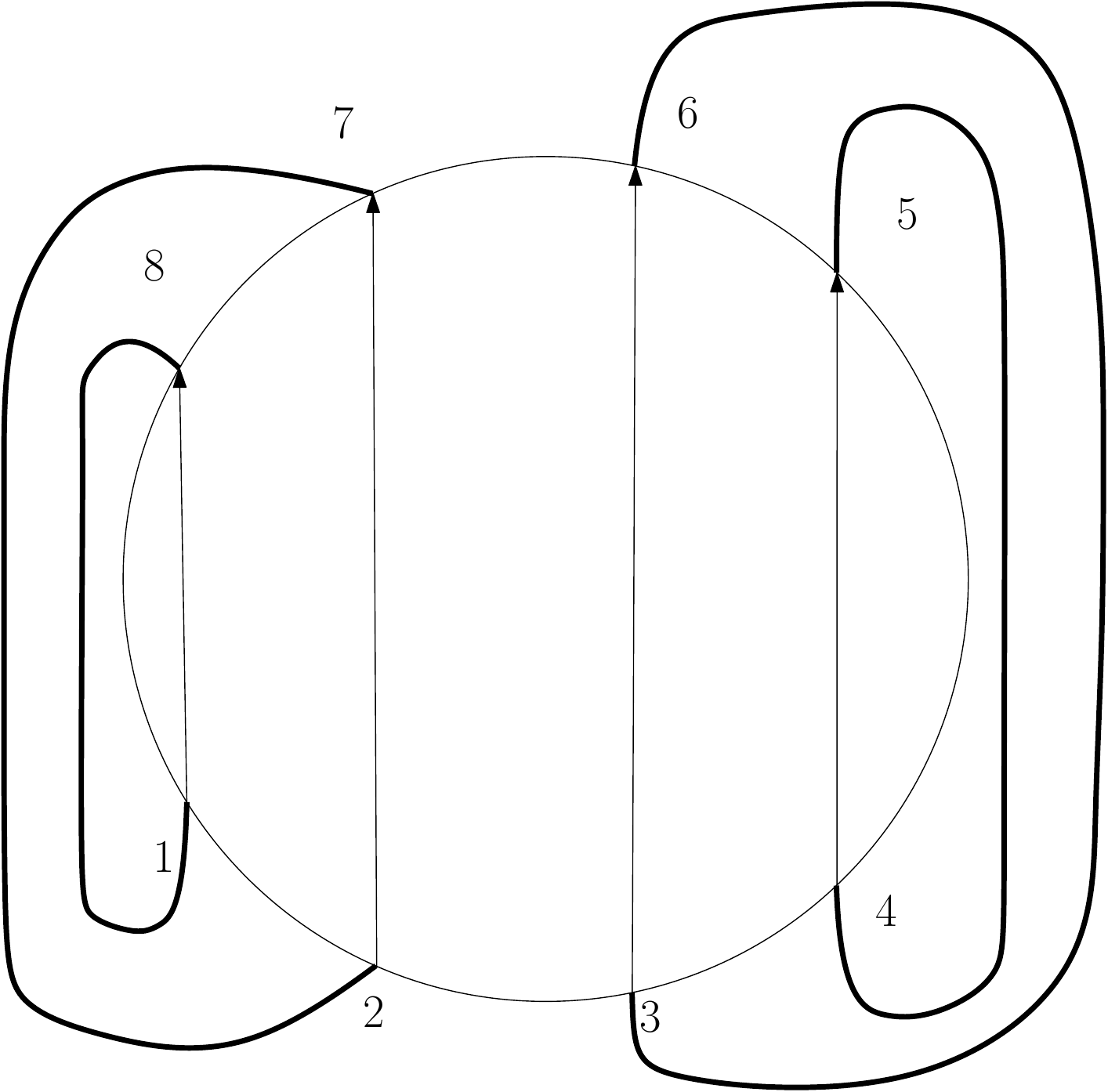}}
  \hspace{10mm}
  \subcaptionbox{Incompatible Pairings}
  {\includegraphics[width=.27\linewidth]{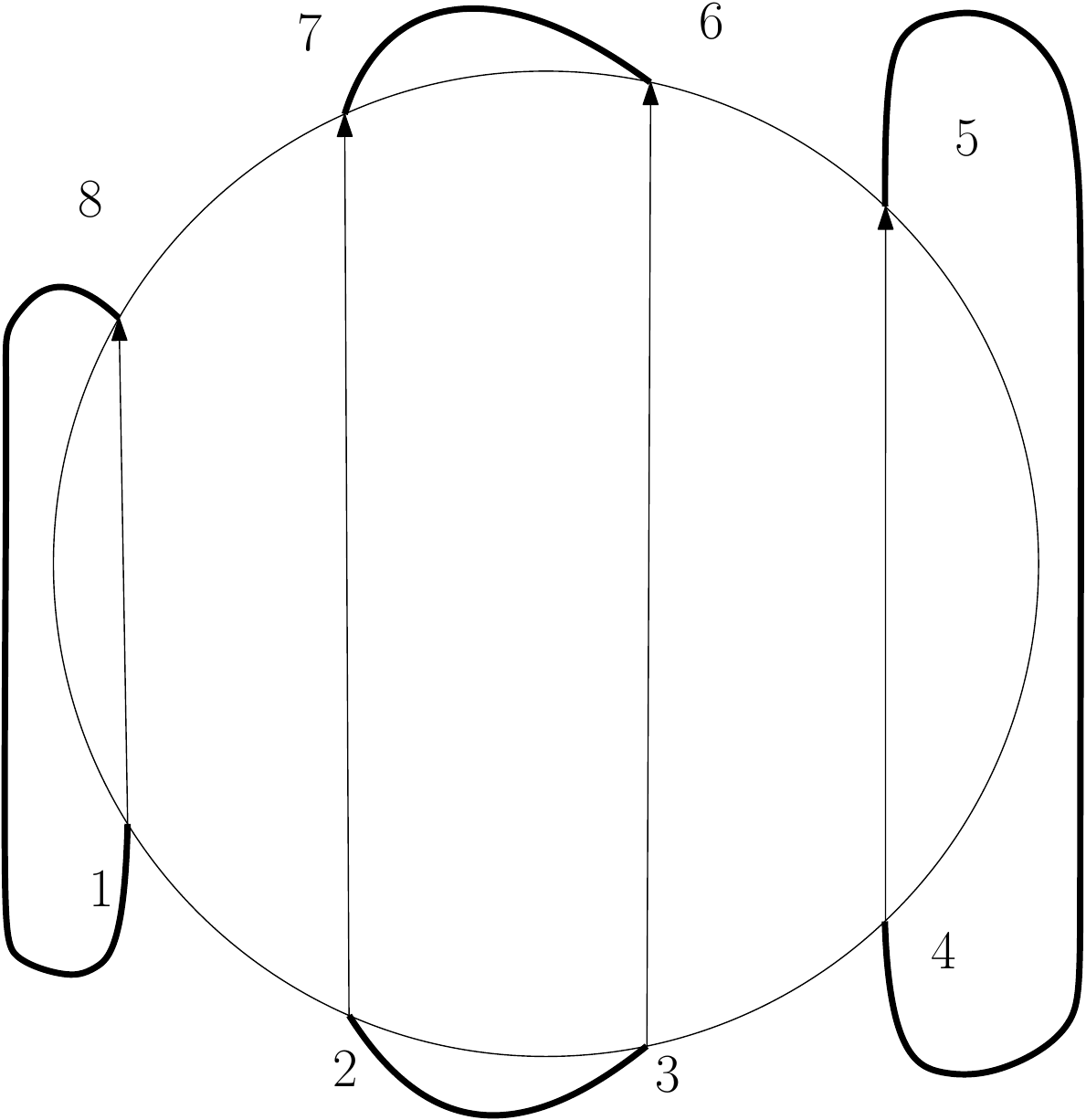}}
\caption{Compatible and Incompatible Pairings}
\label{fig:configs}
\end{figure}
See Figure~\ref{fig:configs} for an example. Let $\ang{M,M'}$ form a compatible pair. 
We call the edges of $M$ as internal edges and those of $M'$ as external edges.
\begin{lemma}\label{comp} 
Compatibility is reflexive and antisymmetric i.e. $\ang{M,M}$ is always 
compatible and if $\ang{M,M'}$ is compatible then $\ang{M',M}$ isn't. 
\end{lemma}
\begin{proof}
$\ang{M,M}$ is always a compatible pair since for any pairing $M$ inside 
just put $M$ outside with demand edges directed in the opposite direction. 
Antisymmetry follows from Lemma~\ref{lem:len}.
\end{proof}
\begin{definition}\label{def:len} 
Define $\mylen{u,v} = \ell(v)-\ell(u)$ for every demand edge $(u,v)$. Let $\mylen{M}$ be the sum of lengths of demand edges of $M$ when the pairing $M$ is placed inside and $\mylen{\vec{M}}$ be the sum of lengths of the demand edges when the pairing comes with directions not necessarily in the standard form.
\end{definition}
For external demand edges $\mylen{u,v}$ may be negative, but for internal edges $\mylen{u,v}$ is positive since the internal demand edges are always drawn with $u < v$.
\begin{lemma}\label{lem:len}
If $\ang{M,M'}$ is a compatible pair and $M \neq M'$ then $\mylen{M} <\mylen{M'}$.
\end{lemma} 
\begin{proof}(of Lemma~\ref{lem:len})
It suffices to prove this for a non-trivial cycle in $M \cup M'$. Let the edges of the cycle $C$ be partitioned into $A,A'$ according to which one is inside. $\mylen{A} +\mylen{\vec{A'}} = 0$ where $\vec{A'}$ is the version of $A'$ oriented according to the orientation of $M'$ when placed outside (because each vertex of $C$ occurs with opposite sign in $\mylen{A}$ and $\mylen{\vec{A'}}$. Notice that to go from $\vec{A'}$ to $A'$ we need to convert the reverse edges to forward edges, which increases the absolute value of $\mylen{\vec{A'}}$). Since in absolute value $A$ and $\vec{A'}$ have the same length, the lemma follows.
\end{proof}
\begin{remark}
This yields an alternative shorter proof of the Parallel Bijection Lemma~\ref{lem:bij} by observing that the parallel pairing is the unique pairing with maximum length thus has no compatible pairing other than itself.
\end{remark}
A set of disjoint paths $R$ in $G$ between a collection of  pairs of terminals which form a pairing $M$ is called a routing. We say that $R$ \emph{corresponds} to $M$ in this case.

For pairings $M,M'$ let $W(\ang{M,M'})$ denote the weighted signed sum of all cycle covers consisting of the pairing $M$ inside the face and routing $R'$ that correspond to the pairing $M'$, outside the face.
\begin{observation}\label{obs:len}
$W(\ang{M,M'})$ will be zero unless $\ang{M,M'}$ is a compatible pair or $M = M'$.
\end{observation}
Also notice that the cycle cover has an arbitrary set of (disjoint) cycles
covering vertices not lying on the routing in the sense that we may cover such
vertices by non self-loops. Let's abbreviate 
$W(\ang{M,*}) = \sum_{M' :M' \mbox{ is a pairing}}{W(\ang{M,M'})}$.
From Lemma~\ref{lem:len} and Observation~\ref{obs:len} we have that:
\begin{proposition}\label{prop:len}
$W(\ang{M,*}) =  \sum_{M' : \mylen{M'} > \mylen{M} \vee M' = M}{W(\ang{M,M'})}$
\end{proposition}
Another fact we will need is that:
\begin{proposition}\label{prop:sign}
$W(\ang{M,M'}) = (-1)^{k-c_{M,M'}}W(\ang{M',M'})$ where $c_{M,M'}$ 
is the number of cycles passing through at least one demand edge
in the union $M \cup M'$ (and $k$ the total number of terminal pairs and 
equals the number of cycles in $\ang{M',M'}$).
\end{proposition}
\begin{proof}(of Proposition~\ref{prop:sign})
Notice that the paths belonging to the routing $R'$ are the same in
both $\ang{M,M'}$ and $\ang{M',M'}$. Thereafter it is an immediate 
consequence of the assumption that the number of cycles in 
$M \cup M'$ is $c_{M,M'} + k'$ (where $k'$ is the the number of cycles 
avoiding all terminals in $\ang{M,M'}$), in $M \cup M$ is $k + k'$ 
(because number of cycles avoiding all terminals in $\ang{M,M'}$ is 
the same as the number of cycles in $\ang{M',M'}$) and of Fact~\ref{fact:innocuous}.
\end{proof}
Thus by plugging in the values from Proposition~\ref{prop:sign} in Proposition~\ref{prop:len} and rearranging, we get the main result 
of this section:
\begin{lemma}[Cancellation Lemma]\label{lem:cancel}
$$W(\ang{M,M}) =  W(\ang{M,*}) + \sum_{M' : \mylen{M'} > \mylen{M}}{(-1)^{k+c_{M,M'}+1}W(\ang{M',M'})}$$
\end{lemma}
We illustrate this with an example in Subsection~\ref{expl}.
\subsection{An Example of the One-face Case}\label{expl}
Let $M_1 = \{(1,8),(2,5),(3,4),(6,7)\}$ be the pairing to start with. $M_1$ is compatible with a routing, say $R_2$, whose corresponding pairing is $M_2 = \{(1,8),(2,7),(3,4),(5,6)\}$. We consider the pairing $M_2$ then which is compatible with another routing, say $R_3$ and the corresponding pairing be $M_3 = \{(1,8),(2,7),(3,6),(4,5)\}$. 
Since $M_3$ is in parallel configuration, from Lemma~\ref{lem:bij} the only routing compatible with $M_3$ corresponds to $M_3$ itself and the recursion stops. We illustrate this in Figure~\ref{fig:test}. From the above discussion, we have the following sequence of equations. 
\begin{figure}
\centering
  \subcaptionbox{}
  {\includegraphics[width=.25\linewidth]{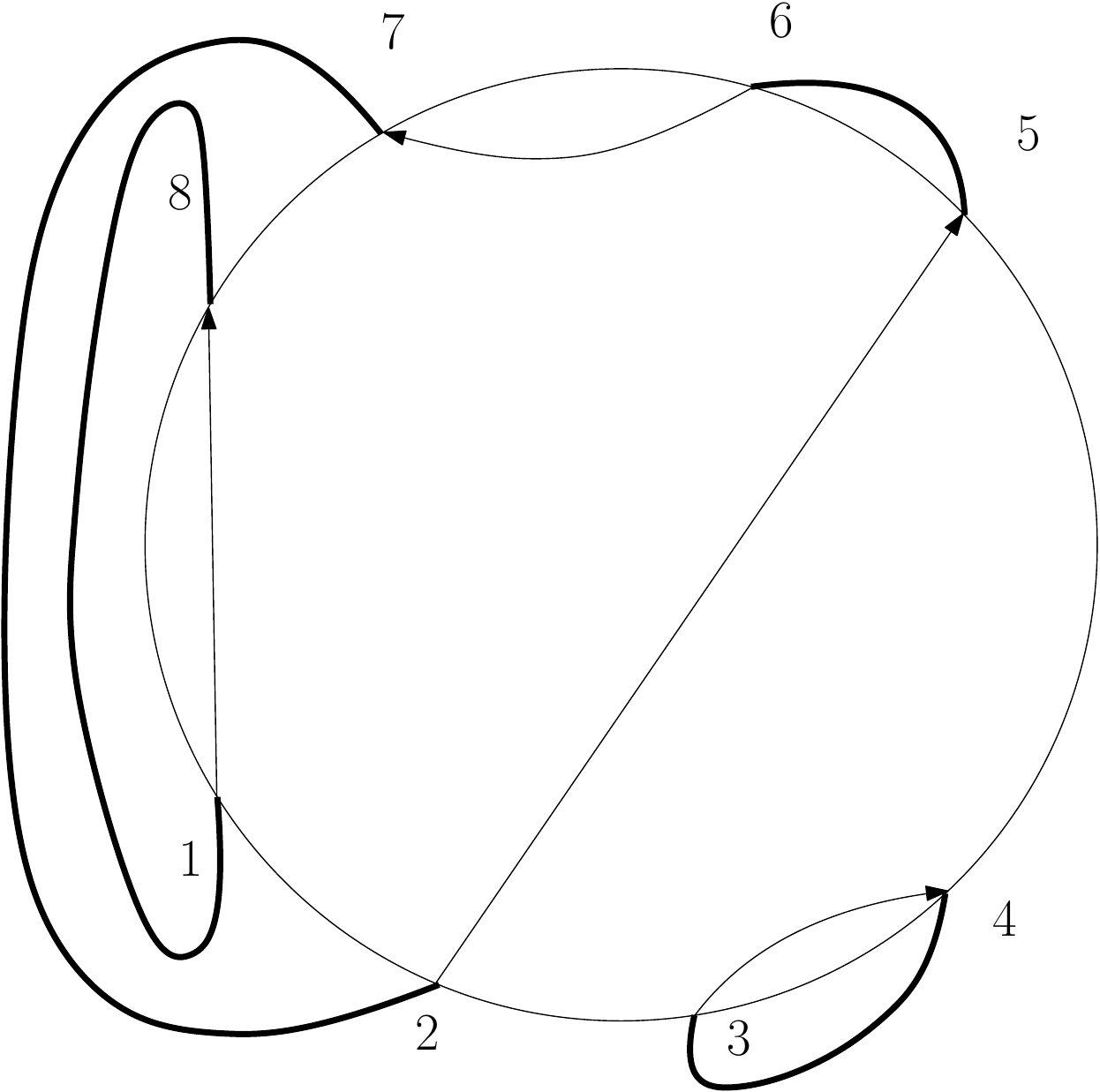}}
  \hfill
  \subcaptionbox{}
  {\includegraphics[width=.3\linewidth]{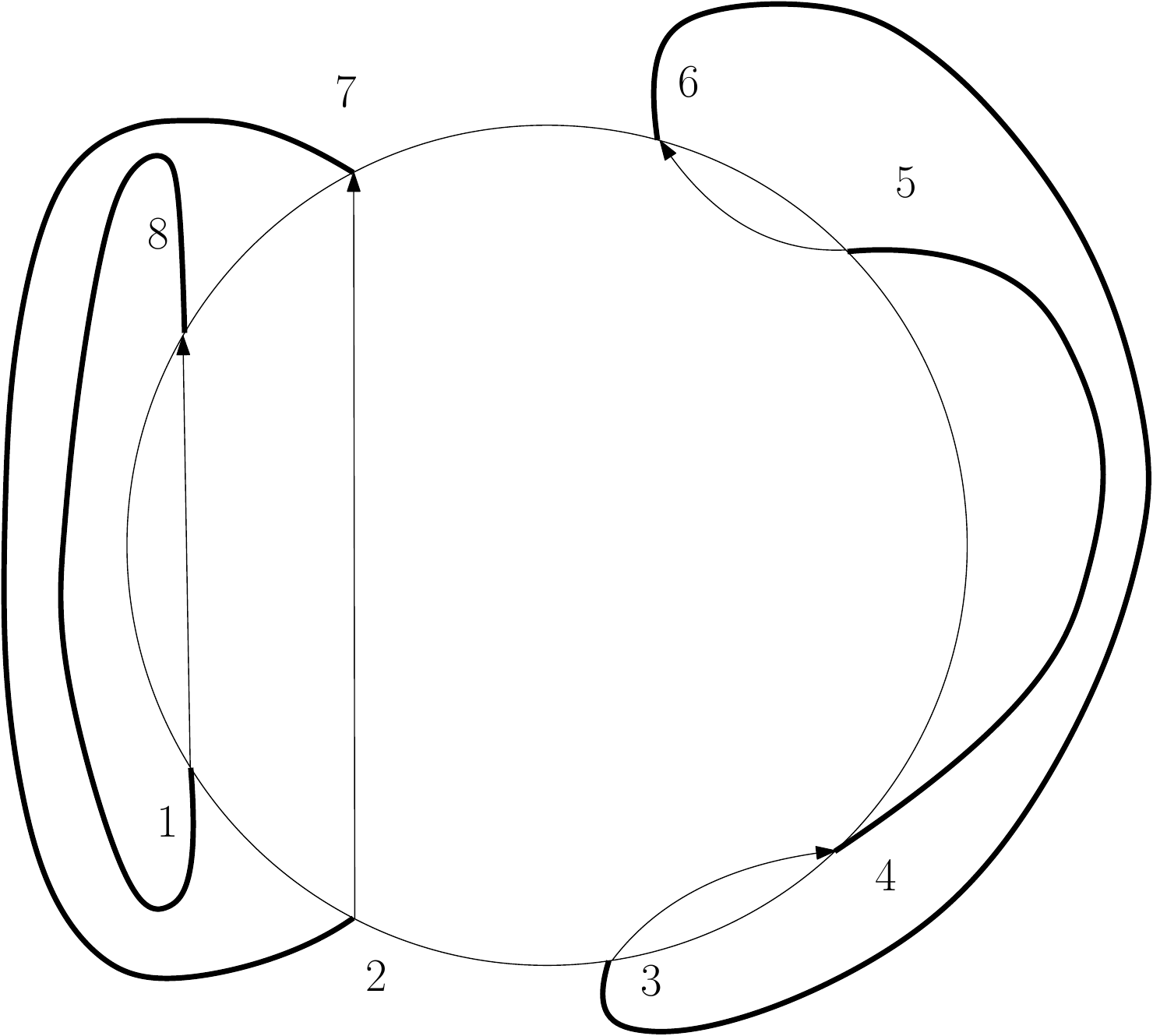}}
  \hfill
  \subcaptionbox{}
  {\includegraphics[width=.25\linewidth]{exmp3.pdf}}
  \hfill
\caption{An Example (a) $M_1 \cup M_2$ (b) $M_2 \cup M_3$ (c) $\ang{M_3,*} = M_3 \cup M_3$ }
\label{fig:test}
\end{figure}
\begin{eqnarray*}
W(\ang{M_1,M_1})  &=& W(\ang{M_1,*}) - W(\ang{M_1,M_2}) \\
W(\ang{M_1,M_2})  &=& - W(\ang{M_2,M_2}) \\
W(\ang{M_2,M_2})  &=& W(\ang{M_2,*}) - W(\ang{M_2,M_3}) \\
W(\ang{M_2,M_3})  &=& - W(\ang{M_3,M_3}) \\
W(\ang{M_3,*}) &=& W(\ang{M_3,M_3})
\end{eqnarray*}
After substitutions we get, 
$$ W(\ang{M_1,M_1}) = W(\ang{M_1,*}) + W(\ang{M_2,*}) + W(\ang{M_3,*}) $$
\newpage
\section{Disjoint Paths on Two faces: The parallel case}\label{sec:twoface}
In this section, we solve the shortest $k$-DPP on planar graphs such that all terminals lie on two faces, say $f_1,f_2$ in some embedding of the graph and all the demands are directed from one face to another. The key difference between the one-face case and the two-face case is that the compatibility relation in the two-face case is not antisymmetric. Consequently, the pairings in the two-face case cannot directly be put together as a DAG(see Figure~\ref{fig:twofacepar}) and we are unable to perform an inclusion-exclusion (like in Lemma~\ref{lem:cancel}).
\begin{figure}
\centering
  \includegraphics[width=.4\linewidth]{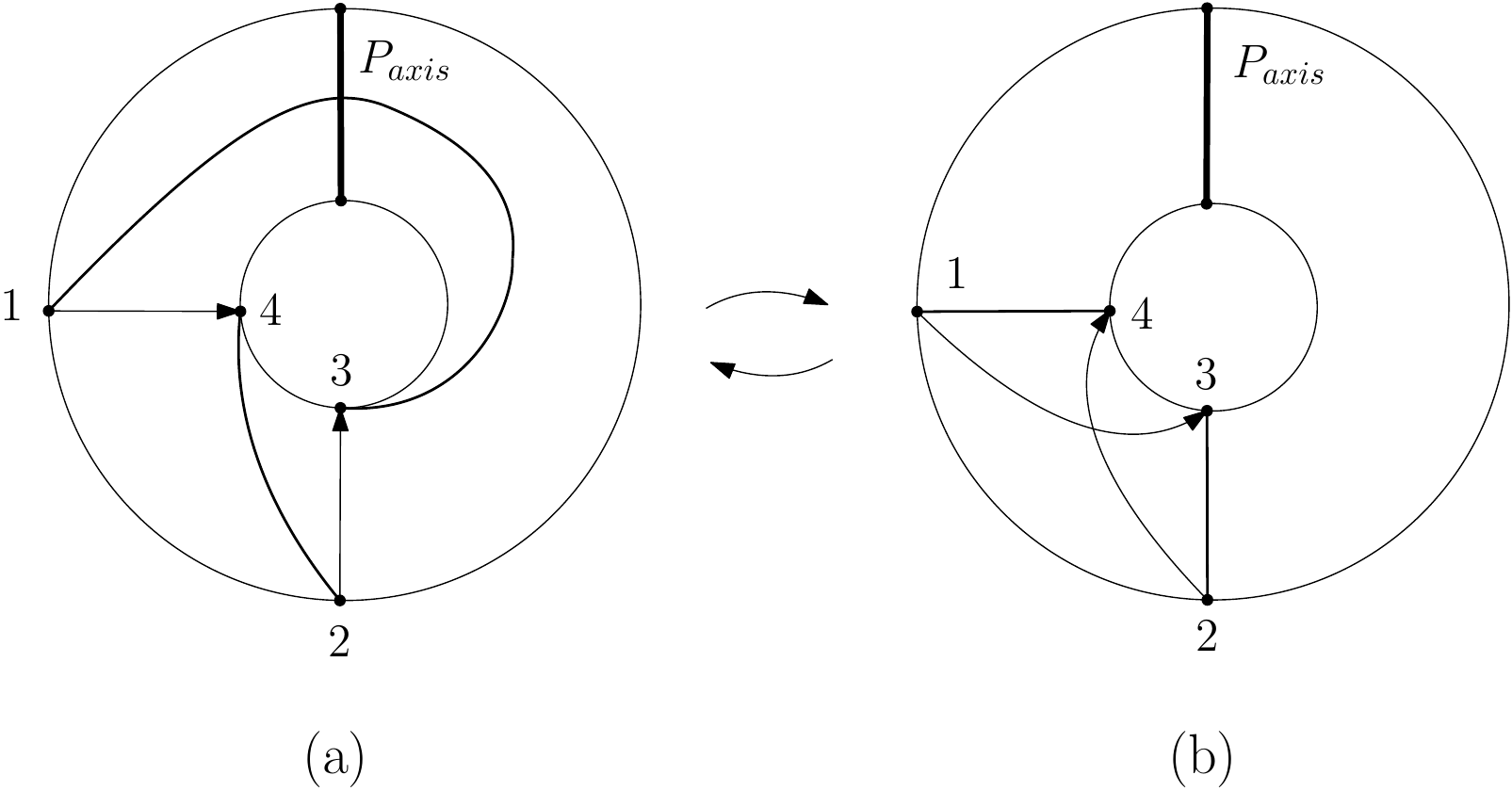}
  \caption{The presence of two faces allows routings of two pairings to be present in the determinant of each other like in this example. $P_{axis}$ is a path between the two faces. (a) shows two parallel demands on two faces and (b) shows a different configuration for the two parallel demands. Notice that one of the two paths necessarily needs to cross the axis in order to obtain (b) from (a), whereas to obtain the pure cycle cover of (a) both paths must cross the axis equal number of times.}
  \label{fig:twofacepar}
\end{figure}
\subparagraph*{Notation and Modification.}
We connect $f_1,f_2$ by a path $P_{axis}$ in the directed dual graph $G^*$. We consider the corresponding primal arcs of $P_{axis}$ which are directed from $f_1$ to $f_2$(in the dual) and weigh them by an indeterminate $y$. Without loss of generality, we can assume that these arcs are counter clockwise as seen from $P_{axis}$. Similarly, the primal arcs of $P_{axis}$ which are directed from $f_2$ to $f_1$(in the dual) are weighed by $y^{-1}$. According to our convention, these arcs are clockwise as seen from $P_{axis}$. We number the terminals of the graph in the following manner. Take the face $f_2$ and start labeling the terminals in a counter-clockwise manner starting from the vertex immediately to the left of $P_{axis}$ as $1,2,\ldots,k$ and then label the terminals of $f_1$ again in a counter-clockwise manner starting from the vertex immediately to the right of the dual path as $k +1,\dots,2k$. 
For any terminal $s$, $\ell(s)$ describes the label associated with $s$. We now apply the modification step in Section~\ref{sec:det} and direct the demand edges forward.
Throughout this section, we fix a pairing $M$ such that each demand edge of $M$ has one terminal on either face. We refer to these types of demand edges as cross demand edges and denote them by $\CD_{M}$. Clearly, $|\CD_M| = k$.
\subsection{Pure Cycle Covers}
Like in Subsection~\ref{sub:pcc0} \emph{pure cycle covers} are defined to be cycle covers $CC$, such that each cycle in $CC$
which contains a terminal also contains the corresponding mate of that terminal and no other terminal.
We begin with Lemma~\ref{lem:intuitive} from \cite{RobertsonS86a} which will be useful to analyze the two-face $k$-DPP. In their notation, the two faces having terminals are $C_1$, $C_2$ with $C_1$ inside $C_2$ in the embedding of $G$. 
For completeness sake, we have provided a proof in the Appendix. 
\begin{lemma}[Quoted from Section~5~\cite{RobertsonS86a}]\label{lem:intuitive}
We represent the surface on which $C_1,C_2$ are drawn by $\sigma = \{(r,\theta): 1 \leq r\leq 2,0\leq\theta\leq 2\pi\}$. Let $f: [0,1] \rightarrow \sigma$ be continuous. Then it has finite \emph{winding number} $\theta(f)$ defined intuitively as $\frac{1}{2\pi}$ times the the total angle turned through (measured counterclockwise) by the line $OX$, where $O$ is the origin, $X= f(x)$, and $x$ ranges from 0 to 1. Let $\mathcal{L}$ be a set of $k$ paths drawn on $\sigma$, pairwise disjoint. We call such a set $\mathcal{L}$ a linkage. If $\mathcal{L}$ is a linkage then clearly $\theta(P)$ is constant for $P \in \mathcal{L}$, and we denote this common value by $\theta(\mathcal{L})$. 
\end{lemma}
We distribute the terminals of the cross demands($\CD_M$) evenly on the faces $f_1$ and $f_2$ at intervals of $\frac{2\pi}{|\CD_M|}$. For convenience sake, assume that the graph is embedded such that $P_{axis}$ is a radial line. Our proofs go through even if this is not the case simply by accounting for the angle between the endpoints of the axis. The other terminals, vertices and edges of $G$ are embedded such that the graph is planar. Claim~\ref{claim:pcc3}, while not being crucial in the analysis, still helps us understand how demand edges occur in the parallel Two-Face case.
\begin{claim}\label{claim:pcc3}
For any three demand edges in $\CD_M$, all three of them cannot interlace with each other.
\end{claim}
\begin{proof}
Assume that the claim does not hold for three demand edges $h_1,h_2,h_3 \in \CD_M$ such that $l(s_1)<l(s_2)<l(s_3)$. Since all three edges interlace, we have that $l(t_1)>l(t_2)>l(t_3)$. If this is the case, we show that $M$ cannot support a pure cycle cover, say $CC$. Let $C_1,C_2,C_3$ be the cycles of $CC$ including the demand edges $h_1,h_2,h_3$ respectively. Since the cycle cover is pure, there exist disjoint paths, say $P_1,P_2,P_3$, between the endpoints of the three demand edges. Also consider the paths $P_4,P_5$ which are comprised of the edges of $f_1$ from $s_1$ to $s_3$ via $s_2$ and $t_1$ to $t_3$ without using $t_2$. Paths $P_1,P_3,P_4,P_5$ form a cycle in the graph with $s_2$ inside and $t_2$ outside it. Therefore, $P_2$ must intersect either $P_1$ or $P_3$ which gives a contradiction.
\end{proof}
We say that a cycle cover $CC$ \textit{effectively crosses the axis} $x$ times if the total number of times the paths in $CC$ cross $P_{axis}$ counter-clockwise is $x$ more than the total number of times they cross it in the clockwise direction. We abbreviate this by $\AxisCross_{M,CC}$. We now characterize pure cycle covers(Lemma~\ref{lem:pcc7}).
\begin{observation}\label{obs:pcc6}
If $P$ is any path(on the plane) in $G$ such that $\theta(P) = 2\pi$ then $P$ effectively crosses the axis exactly once in the counter-clockwise direction. 
\end{observation}
\begin{proof}(Sketch)
We know that $\theta$ is a continuous function and its evaluations at the start and end of $P$ are zero and $2\pi$ respectively. By the intermediate value theorem, it follows that on some point of $P$, $\theta$ takes on the value $\theta_0$ where $\theta_0$ which is the angle between the start of $P$ and any point on $P_{axis}$. Since the direction of measurement is counter-clockwise, we conclude that $P$ must cross $P_{axis}$ exactly once in the counter-clockwise direction. 
\end{proof}
\begin{lemma}\label{lem:pcc7}
Assuming $\CD_M \neq \emptyset$, for any pure cycle cover $CC$, $\AxisCross_{M,CC} = \omega|\CD_M| + \Offset_M$ for some integer $\omega$ and a fixed integer $\Offset_M \in \{0,1,\ldots,|CD_M|-1\}$. 
\end{lemma}
\begin{proof}
We only have to show that the cross demands must contribute to $\AxisCross_{M,CC}$ by an amount of $\omega|\CD_M| + \Offset_M$. 
As $CC$ is a pure cycle cover, we know from Lemma~\ref{lem:intuitive} that each path between a terminal pair traverses the same angle, say $\theta = 2\pi\omega$ for some integer $\omega$. Since each path traverses the same angle, each source terminal is routed to its corresponding sink terminal which is shifted by an angle of $\theta_0 \in [0,2\pi)$ and therefore, $\theta_0$ can be written as $2\pi\frac{\Offset_M}{\CD_M}$ where $\Offset_M \in \{0,1,\ldots,|\CD_M-1|\}$ is the common offset. Observe that the offset is dependent only on the pairing $M$ and is not related to the cycle cover. 
Summing this angle for all demand edges in $\CD_M$, the total angle traversed by the corresponding paths in $CC$ is simply $\theta|\CD_M| = 2\pi\omega|\CD_M| + 2\pi\Offset_M$. From Observation~\ref{obs:pcc6} every time an angle of $2\pi$ is covered, we effectively cross the axis exactly once. Thus the value of \AxisCross~due to the demands in $\CD_M$ is $\omega|\CD_M| + \Offset_M$.
\end{proof}
\subsection{Pruning Bad Cycle Covers}
As a consequence of the topology of the One-Face case, the compatibility relation for pairings is antisymmetric and therefore a straightforward inclusion-exclusion is enough to cancel all the ``bad'' cycle covers. In the two face case, there may exist a set of compatible pairings which yield routings of each other in the determinant, thus making it impossible to cancel ``bad'' cycle covers. Therefore, we must make distinction between compatible pairings which yield pure cycle covers and the ones which yield ``bad'' cycle covers.
\begin{definition}[Compatibility \& \M-Compatibility]\label{TColl}
Consider two pairings $M,M'$. We say that $M'$ is compatible with $M$ if there exists a routing $R'$ yielding a pure cycle cover for $M'$, which when combined with the demand edges of $M$, forms a cycle cover, denoted by $CC_{R'}$. Moreover, if $CC_{R'}$ satisfies the following property, we say $M'$ is \M-compatible for $M$.
\begin{equation}
\tag{Modular Property}
\AxisCross_{M,CC_{R'}} \equiv \Offset_M(\bmod\mbox{ }|\CD_M|)
\label{MmodP}
\end{equation}
\end{definition}
From Lemma~\ref{lem:pcc7}, it is clear that $M$ is \M-compatible with itself. We now show that any other $M' \neq M$ is not \M-compatible with $M$.
\begin{lemma}\label{lem:AxisCross}
For any routing $R'$ corresponding to a pairing $M'$ such that $M'\neq M$, $$\AxisCross_{M,CC_{R'}} \not\equiv \Offset_M(\bmod\mbox{ }|\CD_M|)$$
\end{lemma}
\begin{proof}
Let $\{P_1,P_2,\ldots,P_{k}\}$ be $k$ disjoint paths in the routing $R'$. Next, we use Lemma ~\ref{lem:intuitive} to say that each path in the set must have the same angle as seen from the center of the concentric faces. Since the routing does not lead to a pure cycle cover of $M$, each source terminal is routed to a sink terminal which is shifted by an angle of $\theta'_0 \in [0,2\pi)$ and therefore, $\theta'_0$ can be written as $2\pi\frac{\Offset_{R'}}{\CD_M}$ where $\Offset_{R'} \in \{0,1,\ldots,|\CD_M-1|\}\backslash\{\Offset_M\}$ is the common offset that each path traverses. Notice that pure cycle covers will have an offset of $\Offset_M \neq \Offset_{R'}$ since in the pure case, the offset between the source and sink must be different from that of the offset of $O_{R'}$, otherwise $R'$ would be a pure cycle cover. Therefore,
\begin{align}
\theta(P_1)&=\theta(P_2)=\ldots=\theta(P_{k})= 2\omega\pi + \theta'_0\\
\implies \theta(\bigcup_{i=1}^{k}P_i) &= 2\pi(\omega.|CD_M| + \Offset_{R'})
\end{align}
From Observation~\ref{obs:pcc6} every time an angle of $2\pi$ is covered, we effectively cross the axis exactly once. Thus the value of \AxisCross\mbox{ } due to the routing $R'$ is $2\omega|\CD_M| + \Offset_{R'}$. Since, $\Offset_{R'} \not\equiv \Offset_M \bmod\mbox{ }|\CD_M|$, we conclude that $R'$ does not satisfy \eqref{MmodP}.
\end{proof}
\begin{theorem}\label{thm:equiLenM}
Let $M,M'$ be two $2$-face pairings such that $M'$ is \M-compatible for $M$. Then it must be the case that $M = M'$.
\end{theorem}
Theorem~\ref{thm:equiLenM} is a consequence of Lemma~\ref{lem:pcc7} and Lemma~\ref{lem:AxisCross}. Now, we describe a computational procedure using which we can obtain only the pure cycle covers of $M$.
\subsection{Determinant from product of matrices}\label{subsec:MV97}
The determinant of an integer matrix is complete for the class \Gl \cite{Damm,Toda,Vinay} and Mahajan-Vinay~\cite{MV} give a particularly elegant proof of this
result by writing the determinant of an $n \times n$ matrix
as the difference of two entries of a product of $n+1$ matrices of size
$2n^2 \times 2n^2$. By a simple modification of their proof we can obtain
each coefficient of the determinant - which is a univariate polynomial  
(or in fact any polynomial with constant variables) - in \Gl. One way to do
so is to evaluate the polynomial at several points and then interpolate.
\begin{remark}\label{rem:twofaceseq}
We can adapt the sequential algorithm for the One-face case as follows. We first compute the determinant. The determinant in this case, is a bivariate polynomial. Since we only care for exponents of $y$ to be modulo $k$, we may evaluate this polynommial in $y$ at all the $k^{th}$ roots of unity. Upon taking their sum, all the monomials whose exponents are not equivalent to $0$ modulo $k$ cancel out. We can divide the resulting polynomial by $k$ to preserve the coeffecients. We can now evaluate the polynomial at $n$ points and then interpolate as we do in the One-Face Case. See Remark~\ref{rem:count}. This gives us the same complexity as in the One-face case, with an additional blow-up of $k$. In order to do this, we need to shift to a model of computation which allows us to approximately evaluate polynomials at imaginary points.
\end{remark}
\noindent\subparagraph*{Computing the univariate polynomial in the two-face case.}
We briefly review the algorithm described by Mahajan and Vinay\cite{MV} to compute the determinant.
Instead of writing down the determinant as a sum of cycle covers, they write it as a sum of clow
sequences. A clow sequence which generalises from a cycle cover allows walks that may
visit vertices many times as opposed to cycles where each vertex is visited exactly once(for more
details see\cite{MV}). Even though the determinant is now written as a sum over many more terms,
they show an involution where any clow sequence which is not a cycle cover cancels out with 
a unique ``mate'' clow sequence which occurs with the opposite sign. In order to implement this 
determinant computation as an
algorithm, each clow which can be realised as a closed walk in the graph is computed in a
non-deterministic manner.

Our only modification to the algorithm is as follows: in each non-deterministic
path, we maintain a $O(log(k))$-bit counter which counts the number of times edges from $P_{axis}$
have been used in the clow sequence so far modulo $k$. In other words, everytime the counts exceeds $k$,
we shift the counter to $0$. At the end of the computation, the number in this counter 
is exactly the exponent of $y$ modulo $k$. It is easy to see that clow sequences which 
are not cycle covers, still
cancel out because, in a clow sequence and its mate the set of directed edges traversed is the same.
Consequently, at the end of the computation of each clow sequence, a clow and its made get the same 
exponent in $y$ modulo $k$. This can be done in \Gl~as described in \cite{MV}. 
\section{Proof of the Main Theorem}\label{sec:mainproof}
We split the proof of the Theorem~\ref{thm:main} into three parts. Theorem~\ref{thm:oldMain} gives a proof of the count and the associated sequential and parallel complexity bounds. In Subsection~\ref{sub:ptime} we describe how to do search in polynomial time. Lastly, in Subsection~\ref{sub:RNC} we show that search can be done in \RNC. Observe that for the decision version of the shortest $k$-DPP, it suffices to check whether the polynomial obtained by the signed sum of determinants is non-zero or not.
\subsection{Proof of Main Theorem: Count in \NCt}
\begin{theorem}\label{thm:oldMain}
Given an undirected planar graph $G$ with $k$ pairs of source and sink 
vertices lying on a single face $F$, then the count of all shortest 
$k$-disjoint paths between the terminals can be found in time
$O(k^2 + \log^2{n})$ parallel time using $4^{O(k)}n^{O(1)}$ processors.
It can also be found in sequential time $O(4^k n^{\omega+1})$. 

When the terminals are distributed such that all the sources are on one face and all the sinks are on the other, the count of all shortest $k$-DPP can be obtained in sequential time $O(kn^{\omega +1})$. 
\end{theorem}
\begin{proof}
In the one face case, the Cancellation Lemma~\ref{lem:cancel} 
allows us to cancel out all cycle covers that are not good
(i.e. those which do not correspond to the input
terminal pairing $M_0$) and replace them by a signed sum of
$W(\ang{M,*})$ for various matchings and also $W(\ang{P,P})$ where 
$P$ is the unique parallel pairing. This replacement can be done
in time linear in the total number of possible terms. 

Observe that there are at most $4^k$ different pairings possible
(since they correspond to outerplanar matchings which are bounded in number
by the Catalan number $\frac{1}{k+1}{{{2k}\choose{k}} < 4^k}$ 
see e.g. \cite{Hernando2002}) We obtain the count itself by evaluating the polynomial at $n$ points followed by interpolation(see Remark~\ref{rem:count}). This accounts for a blow-up of $n$ in the sequential running time. 
Notice that we assume that the determinant of an $n \times n$ matrix can
    be computed in matrix multiplication time $n^\omega$ \cite{AHU}.

Alternatively pairings can be indexed by $k$-bit
numbers. We can build a matrix indexed by $M,M'$ and containing
zero if $\ang{M,M'}$ is not a compatible pair and the 
 sign with which $W(\ang{M',M'})$ occurs in the expression for 
$W(\ang{M,*})$, otherwise.
The matrix is of size $4^k \times 4^k$. This represents an 
system of equations $Cx = b$ (where $C$ is the compatibility
matrix above and entries of column vector $b$ are $W(\ang{M,*})$.   
Notice that the system is upper triangular 
because $\mylen{M} < \mylen{M'}$ for compatible $\ang{M,M'}$. 
Also along the diagonal we have $\pm 1$'s
because $W(\ang{M,M})$ always occurs in the expression for $W(M,*)$.
Thus the determinant of $C$ is $\pm 1$. We can invert the matrix in \NC.

In the two face case, instead of directly computing the determinant, we use the computation described in Subsection~\ref{subsec:MV97}. Remark~\ref{rem:twofaceseq} tells us that we can obtain a polynomial such that the exponent in the variable $y$ is at most $k-1$. In this polynomial we look for the terms whose exponent in $y$ is equal to $\Offset_{M}$ and among these terms we extract the monomial with the smallest exponent in $x$ to obtain the shortest pure cycle covers.
\end{proof}
\begin{remark}\label{rem:count}
In order to obtain the count, we need to use polynomial interpolation. The following fact together with Theorem~\ref{thm:oldMain} establishes that count can be obtained in \NCt.
\end{remark}
\begin{fact}\label{fact:interp}[Folklore \cite{CaiSivakumar,Tzameret}]
Polynomial interpolation i.e. obtaining the coefficients of a univariate polynomial given
its value at sufficiently many (i.e. degree plus one) points is in $\TCz \subseteq \NCo$.
It is also in $O(n\log{n})$ time (where $n$ is the degree of the polynomial) via Fast Fourier Transform.
\end{fact}
\subsection{Proof of Main Theorem: Search in \Pt}\label{sub:ptime}
Let $C_{tot}$ be the count of total number of shortest $k$-disjoint 
paths in $G$.  For every edge $e \in G$ we remove $e$ and count the 
remaining number of shortest $k$-disjoint paths using the 
sequential counting procedure 
as oracle.
Let $C_{\bar{e}}$ be this count. If $C_{\bar{e}} > 0$, we proceed with 
the graph $G \setminus e$ since the graph still has a shortest 
$k$-disjoint path. If $C_{\bar{e}} = 0$ then every existing shortest 
$k$-disjoint paths contains the edge $e$ so keep $e$ in $G$ and 
proceed with the next edge. Let $H$ be the final graph obtained.
\begin{lemma}\label{lem:seq}
The graph $H$ is a valid shortest $k$-disjoint path.
\end{lemma}
\begin{proof}
It is easy to see that all the edges in $H$ are
part of a shortest $k$-disjoint path. To see that all the 
edges are part of a single shortest $k$-disjoint paths since 
otherwise we could remove that edge, say $e^*$ and will have 
$C_{\bar{e^*}} > 0$ in $H$ and therefore also in the graph $G$
at the time $e^*$ was under consideration contradicting that
$e^*$ was retained.
Since for each edge we spend $O(4^kn^{\omega+1})$ time, the total search time is
$O(4^kn^{\omega+2})$.
\end{proof}
\begin{remark}\label{rem:weighted}
Our algorithm also works for weighted graphs 
where each edge $e$ is assigned a weight $w(e)$ which is 
polynomially bounded in the number of vertices. 
This can be done by putting odd (additive) weights 
$w'(e) = (|E|+1)w(e) + 1$
on the edges i.e. replacing the entry corresponding to $e$ in the
adjacency matrix by $x^{w'(e)}$ instead of just $x$.
Notice that the length of a collection of edges has the same parity
as the sum of its weights. So the calculation in Lemma~\ref{lem:shortGood} go 
through with small changes. This implies that we do not have
to convert a weighted graph into unweighted one in order to run the
counting algorithms and we get the sum of the (additive) weights of
edges instead of counts as a result.
\end{remark}
\subsection{Proof of Main Theorem: Search in \RNC}\label{sub:RNC}
For the construction of shortest $k$-DPP we use the following Isolation lemma introduced by Mulmuley, Vazirani, and Vazirani \cite{MVV}. It is a simple but powerful lemma that crucially uses randomness:
\begin{lemma}[Isolation Lemma] \label{lem:Iso}\sm{Unused}
Given a non-empty $\mathcal{F} \subseteq 2^{[m]},$ if one assigns for each $i \in [m],$ $w_i \in [2m]$ uniformly at random then with probability at least half, the minimum weight subset of in $\mathcal{F}$ is unique; where the weight of a subset $S$ is $\sum_{i \in S} w_i.$
\end{lemma}
\begin{lemma}\label{lem:RNC}
Construction of a solution to the shortest One-face and Two-face Parallel $k$-DPP is in \RNC.
\end{lemma}
\begin{proof}
First we introduce small random weights in the lower order bits of the
edges of the graph $G$ (i.e. give weights like $4n^2 + r_e$ to edge $e$).
Using Lemma~\ref{lem:Iso} these are isolating for the set of $k$-disjoint paths
between the designated vertices, with high probability. In other words
the coefficient of least degree monomial equals $\pm 1$ in the isolating 
case. At the same time the ordering of unequal weight paths is preserved.
This is because the sum of the lower order bits cannot interfere with the
higher order bits of the monomial which represent the length of the 
corresponding $k$-disjoint path.

Let the monomial with minimum exponent be $x^w$. Our counting algorithms 
works for the weighted case as explained in 
the remark in Subsection~\ref{sub:ptime} above. Thus borrowing notation from 
Subsection~\ref{sub:ptime} we can compute $C_{\bar{e}}$ in parallel for 
each edge under the small random weights above. If the weight is indeed
isolating, we will obtain the least degree monomial in $C_{\bar{e}}$ will
be $x^w$ exactly when $e$ does not belong to the isolated shortest $k$-disjoint paths. 
Thus with probability at least half we will obtain a set of shortest $k$-disjoint
paths. When the assignment is not isolating the set of edges which lie on some
shortest $k$-disjoint path will not form a $k$-disjoint path itself so we
will know for sure that the random assignment was not isolating.

We can also give a randomised sequential algorithm for the problem running in
time $O(4^kn^{\omega})$ which uses the idea of inverting a matrix in order to
find a witness for perfect matching described in \cite{MVV}. They use it in the
parallel setting but we apply it in the sequential case also. Essentially we
need to compute all the $O(n)$ many $C_{\bar{e}}$'s in $O(n^\omega)$ time.
Notice that $C - C_{\bar{e}}$ will be precisely the weighted count for the
$k$-disjoint paths that contain the edge $e$. This is precisely the co-factor
of the entry $(u,v)$ where $e = (u,v)$ and since all co-factors can be
computed in $O(n^\omega)$ time we are done.
\end{proof}
\section{Edge disjoint paths}\label{sec:edpp}
\begin{figure}
\centering
  \subcaptionbox{(1): Arbitrary degree to degree $\le 4$, (2): degree $4$ to degree $3$}
  {\includegraphics[width=.4\linewidth]{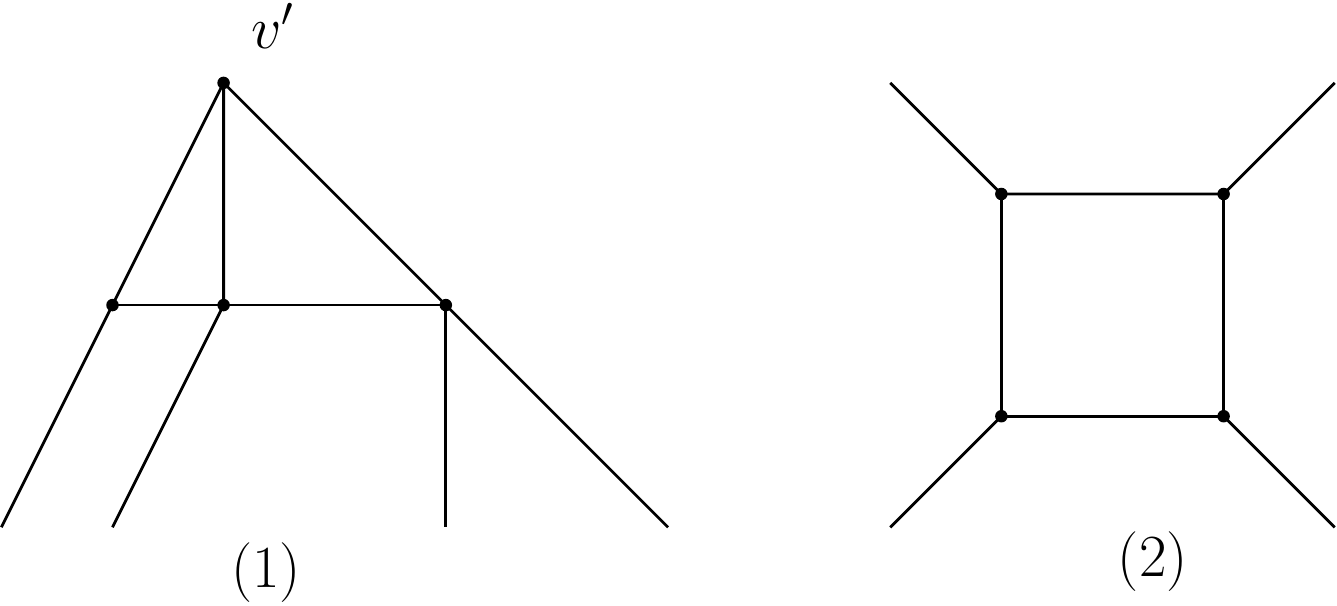}}
  \hspace{5mm}
  \subcaptionbox{Terminal degree Reduction}
  {\includegraphics[width=.4\linewidth]{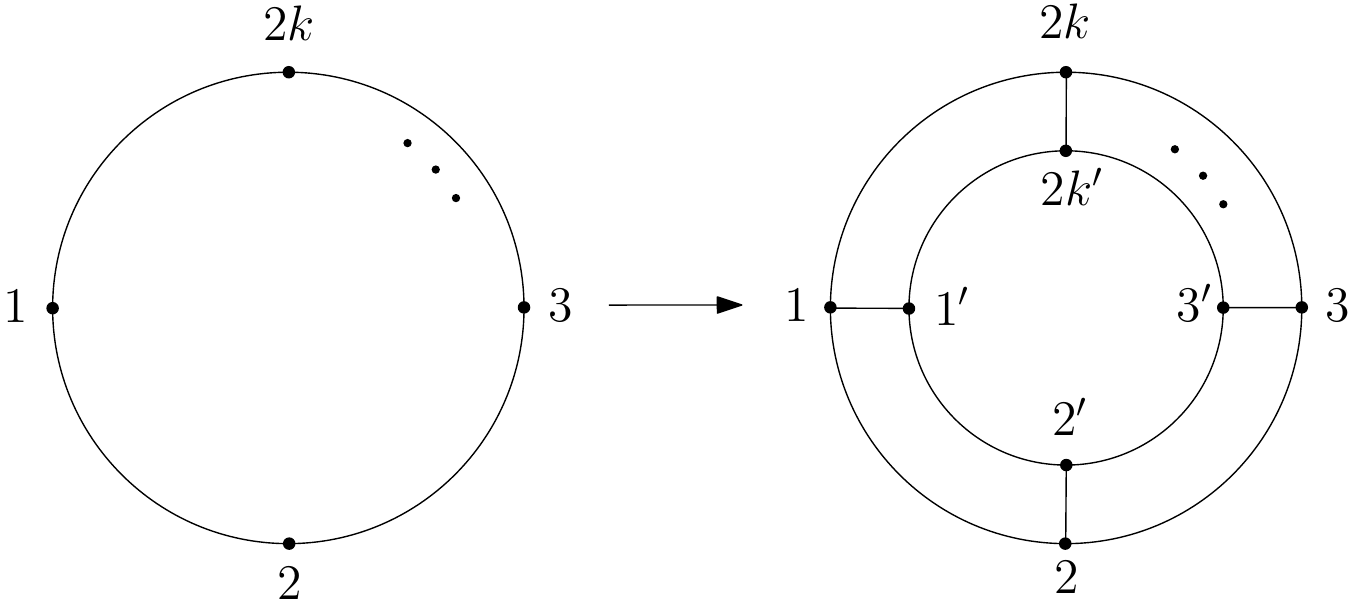}}
  \caption{Degree Reduction Gadgets}\label{fig:edg}
\end{figure}
We define \textit{planar} $k$-EDPP to be the problem of finding $k$ edge disjoint paths in a planar graph $G$ between terminal pairs when, the demand edges can be \textit{embedded} in $G$ such that planarity is preserved. 
We show how to transfer results for $k$ vertex disjoint paths to $k$ edge disjoint paths in undirected graphs using
gadgets in Figure~\ref{fig:edg} borrowed from \cite{MiddendorfP93}.
\begin{lemma}~\label{lem:reduction}
Decision and Search for One-Face \textit{planar} $k$-EDPP reduces to One-Face $k$-DPP. 
\end{lemma}
\begin{proof} (Sketch) The reduction is performed in three steps. First 
    we reduce the degrees of terminals by using the gadget in 
    Figure~\ref{fig:edg}(b)\sm{make this figure} to at most three. 
    Next, we use the gadget in Figure~\ref{fig:edg}(a)(1) to reduce the 
    degree of any vertex which is not a terminal to at most four. After 
    each application of this gadget the degree of the vertex reduces by one.
    A parallel implementation of this procedure would first expand every
    vertex into an, at most ternary tree and then replace each node by the 
    gadget. We then reduce the degrees to at most three by using the 
    gadget in Figure~\ref{fig:edg}(a)(2). Notice that since the demand edges can
    be embedded in a planar manner on the designated face, the disjoint paths 
    can only cross each other an even number of times and hence the for every shortest EDPP
    we will always be able to find a corresponding shortest DPP after using the gadget in
    Figure~\ref{fig:edg}(a)(2). It must also be noted that path lengths will not be preserved, 
    however, we can give any new edges 
    introduced in the gadgets zero additive weight. This can be 
    achieved by simply not weighing the new edges by the indeterminate $x$ 
    in the graph modification step.
    Finally, observe that two paths in a graph with maximum degree three are vertex 
    disjoint iff they are edge disjoint.
\end{proof}
\begin{remark}
    Since counts are not preserved in the gadget reduction, we do not have 
    an \NC-bound for counting $k$-EDPP's. 
\end{remark} 
\section{Conclusion and Open Ends}\label{sec:open}
We have reduced some planar versions of the shortest $k$-DPP to computing determinants. 
This is a new technique for this problem as far as we know 
and has the advantage of being simple and parallelisable while remaining sequentially competitive.

Is it possible to solve the Two-face case with an arbitrary distribution of the demand edges while obtaining similar complexity bounds? The more general question of extending our result to the case when the terminals are on some fixed $f$ many faces also remains open.  
For the One-face case, can we make the dependence on $k$ from exponential to polynomial or even quasipolynomial? 
Also, what about extending our result to planar graphs or even $K_{3,3}$-free or $K_5$ free graphs or to graphs on surfaces. Can one de-randomize our algorithm to get deterministic \NC\ bound for the construction? It will be interesting if one can show lower bounds or hardness results for these problems.
\section*{Acknowledgements}
The first and the fourth authors were partially funded by a grant from Infosys foundation. The second author was an intern at the Chennai Mathematical Institute under the supervision of the first author while some part of this work was done. The fourth author was partially supported by a TCS PhD fellowship.
\bibliography{bibfile}

\begin{thebibliography}{10}

\bibitem{AHU}
Alfred~V. Aho, John~E. Hopcroft, and Jeffrey~D. Ullman.
\newblock {\em The Design and Analysis of Computer Algorithms}.
\newblock Addison-Wesley, 1974.

\bibitem{AM04}
Eric Allender and Meena Mahajan.
\newblock The complexity of planarity testing.
\newblock {\em Inf. Comput.}, 189(1):117--134, 2004.

\bibitem{BH14}
Andreas Bj{\"{o}}rklund and Thore Husfeldt.
\newblock Shortest two disjoint paths in polynomial time.
\newblock In {\em Automata, Languages, and Programming - 41st International
  Colloquium, {ICALP} 2014, Copenhagen, Denmark, July 8-11, 2014, Proceedings,
  Part {I}}, pages 211--222, 2014.

\bibitem{BNZ}
Glencora Borradaile, Amir Nayyeri, and Farzad Zafarani.
\newblock Towards single face shortest vertex-disjoint paths in undirected
  planar graphs.
\newblock In {\em Algorithms - {ESA} 2015 - 23rd Annual European Symposium,
  Patras, Greece, September 14-16, 2015, Proceedings}, pages 227--238, 2015.

\bibitem{CaiSivakumar}
Jin{-}yi Cai and D.~Sivakumar.
\newblock Resolution of hartmanis' conjecture for nl-hard sparse sets.
\newblock {\em Theor. Comput. Sci.}, 240(2):257--269, 2000.

\bibitem{ChuzhoyKL16}
Julia Chuzhoy, David H.~K. Kim, and Shi Li.
\newblock Improved approximation for node-disjoint paths in planar graphs.
\newblock In {\em Proceedings of the Forty-eighth Annual ACM Symposium on
  Theory of Computing}, STOC 2016, pages 556--569, New York, NY, USA, 2016.
  ACM.

\bibitem{ChuzhoyKN16}
Julia Chuzhoy, David H.~K. Kim, and Rachit Nimavat.
\newblock New hardness results for routing on disjoint paths.
\newblock In {\em Proceedings of the 49th Annual ACM SIGACT Symposium on Theory
  of Computing}, STOC 2017, pages 86--99, New York, NY, USA, 2017. ACM.

\bibitem{Damm}
C.~Damm.
\newblock \mbox{DET=L$^{({\rm \#L})}$}.
\newblock Technical Report Informatik--Preprint 8, Fachbereich Informatik der
  Humboldt--Universit{\"a}t zu Berlin, 1991.

\bibitem{DP11}
Samir Datta and Gautam Prakriya.
\newblock Planarity testing revisited.
\newblock In {\em Theory and Applications of Models of Computation - 8th Annual
  Conference, {TAMC} 2011, Tokyo, Japan, May 23-25, 2011. Proceedings}, pages
  540--551, 2011.

\bibitem{VS11}
{\'{E}}ric~Colin de~Verdi{\`{e}}re and Alexander Schrijver.
\newblock Shortest vertex-disjoint two-face paths in planar graphs.
\newblock {\em {ACM} Transactions on Algorithms}, 7(2):19, 2011.

\bibitem{Eil}
Tali Eilam-Tzoreff.
\newblock The disjoint shortest paths problem.
\newblock {\em Discrete Applied Mathematics}, 85(2):113 -- 138, 1998.

\bibitem{FHW80}
Steven Fortune, John~E. Hopcroft, and James Wyllie.
\newblock The directed subgraph homeomorphism problem.
\newblock {\em Theor. Comput. Sci.}, pages 111--121, 1980.

\bibitem{Hernando2002}
C.~Hernando, F.~Hurtado, and Marc Noy.
\newblock Graphs of non-crossing perfect matchings.
\newblock {\em Graphs and Combinatorics}, 18(3):517--532, 2002.

\bibitem{RK}
Richard~M. Karp.
\newblock On the computational complexity of combinatorial problems.
\newblock {\em Networks}, 5:45--68, 1975.

\bibitem{KS}
Yusuke Kobayashi and Christian Sommer.
\newblock On shortest disjoint paths in planar graphs.
\newblock {\em Discret. Optim.}, 7(4):234--245, November 2010.

\bibitem{L75}
James~F. Lynch.
\newblock The equivalence of theorem proving and the interconnection problem.
\newblock {\em SIGDA Newsl.}, 5(3):31--36, September 1975.

\bibitem{MV}
Meena Mahajan and V.~Vinay.
\newblock Determinant: Combinatorics, algorithms, and complexity.
\newblock {\em Chicago J. Theor. Comput. Sci.}, 1997, 1997.

\bibitem{MiddendorfP93}
Matthias Middendorf and Frank Pfeiffer.
\newblock On the complexity of the disjoint paths problems.
\newblock {\em Combinatorica}, 13(1):97--107, 1993.

\bibitem{MVV}
Ketan Mulmuley, Umesh~V. Vazirani, and Vijay~V. Vazirani.
\newblock Matching is as easy as matrix inversion.
\newblock {\em Combinatorica}, 7(1):105--113, 1987.

\bibitem{RobertsonS86a}
Neil Robertson and Paul~D. Seymour.
\newblock Graph minors. {VI.} disjoint paths across a disc.
\newblock {\em J. Comb. Theory, Ser. {B}}, 41(1):115--138, 1986.

\bibitem{RobertsonS88}
Neil Robertson and Paul~D. Seymour.
\newblock Graph minors. {VII.} disjoint paths on a surface.
\newblock {\em J. Comb. Theory, Ser. {B}}, 45(2):212--254, 1988.

\bibitem{RS}
Neil Robertson and Paul~D. Seymour.
\newblock Graph minors. {XIII.} the disjoint paths problem.
\newblock {\em J. Comb. Theory, Ser. {B}}, 63(1):65 -- 110, 1995.

\bibitem{Sch94a}
Petra Scheffler.
\newblock A practical linear time algorithm for disjoint paths in graphs with
  bounded tree-width.
\newblock {\em Technical Report 396}, 1994.

\bibitem{Sch94}
Alexander Schrijver.
\newblock Finding k disjoint paths in a directed planar graph.
\newblock {\em {SIAM} J. Comput.}, 23(4):780--788, 1994.

\bibitem{Schwarzler09}
Werner Schw{\"{a}}rzler.
\newblock On the complexity of the planar edge-disjoint paths problem with
  terminals on the outer boundary.
\newblock {\em Combinatorica}, 29(1):121--126, 2009.

\bibitem{SAN90}
Hitoshi Suzuki, Takehiro Akama, and Takao Nishizeki.
\newblock Finding steiner forests in planar graphs.
\newblock In {\em Proceedings of the First Annual ACM-SIAM Symposium on
  Discrete Algorithms}, SODA '90, pages 444--453, Philadelphia, PA, USA, 1990.
  SIAM.

\bibitem{SYN90}
Hitoshi Suzuki, Chiseko Yamanaka, and Takao Nishizeki.
\newblock Parallel algorithms for finding steiner forests in planar graphs.
\newblock In {\em Proceedings of the International Symposium on Algorithms},
  SIGAL '90, pages 458--467, London, UK, UK, 1990. Springer-Verlag.

\bibitem{Toda}
S.~Toda.
\newblock Counting problems computationally equivalent to the determinant.
\newblock Technical Report CSIM 91-07, Dept of Comp Sc \& Information
  Mathematics, Univ of Electro-Communications, Chofu-shi, Tokyo, 1991.

\bibitem{Tzameret}
I.~Tzameret.
\newblock {\em Studies in Algebraic and Propositional Proof Complexity}.
\newblock PhD thesis, Tel Aviv University, 2008.

\bibitem{vd02}
H.~van~der Holst and J.C. de~Pina.
\newblock Length-bounded disjoint paths in planar graphs.
\newblock {\em Discrete Applied Mathematics}, 120(1):251 -- 261, 2002.

\bibitem{Vinay}
V.~Vinay.
\newblock Counting auxiliary pushdown automata.
\newblock In {\em Proceedings of the Sixth Annual Structure in Complexity
  Theory Conference, Chicago, Illinois, USA, June 30 - July 3, 1991}, pages
  270--284, 1991.

\end{thebibliography}
\appendix
\section{Proof of Lemma~\ref{lem:intuitive}}
\begin{proof}
Recall, the surface on which $C_1,C_2$ are drawn is given by
\[
\sigma = \{(r,\theta): 1 \leq r\leq 2,0\leq\theta\leq 2\pi\}
\]
We quote from \cite{RobertsonS86a}. If $P$ is a path drawn on $\sigma$ with one end in $C_1$ and the other
in $C_2$, let $f: [0,1] \rightarrow \sigma$ be a continuous injection with image $P$ and with
$f(0) \in C_1$, $f(1) \in C_2$; then we define $\theta(P) = \theta(f)$. It is easy to see that this
definition is independent of the choice of $f$. If $P_1, P_2$ are both paths drawn
on $\sigma$ from some $s \in C_1$, to some $t \in C_2$, then $\theta(P_1) - \theta(P_2)$ is an integer, and
is zero if and only if $P_1$ is homotopic to $P_2$. Let $k > 0$ be some fixed integer, and let 
\[
M_i = \{(i,\frac{2j}{k}\pi):1 \leq j \leq k \} (i=1,2).
\]
If $\mathcal{L}$ is a linkage then clearly $\theta(P)$ is constant for $P \in \mathcal{L}$, and we
denote this common value by $\theta(\mathcal{L})$. \\
Intuitively, this is because if any 2 simple paths wind around a face a different number of times then they both must intersect.
\end{proof}
\end{document}